\documentclass[]{elsarticle}

\usepackage[utf8]{inputenc}
\usepackage[T1]{fontenc}

\usepackage[english]{babel}

\usepackage{graphicx}
\usepackage{amsmath}
\usepackage{amsthm}
\usepackage{mathtools}

\usepackage{amssymb}
\usepackage{faktor}
\usepackage{tikz-cd}
\usepackage{comment}

\usepackage{todonotes}

\usepackage{hyperref}
\hypersetup{
	colorlinks,
	linkcolor={red!50!black},
	citecolor={blue!50!black},
	urlcolor={blue!80!black}
}

\usepackage[linesnumbered,onelanguage,algoruled,noend]{algorithm2e}

\newcommand{\F}{\mathbb{F}}
\newcommand{\N}{\mathbb{N}}
\renewcommand{\P}{\mathbb{P}}

\newcommand{\R}{\mathbb{R}}
\newcommand{\Z}{\mathbb{Z}}
\newcommand{\Zpos}{\Z_{>0}}

\newcommand{\bfR}{\boldsymbol{\mathfrak{R}}}

\DeclareMathOperator{\CRT}{CRT}

\newcommand{\SRN}{\mathsf{SRN}}

\newcommand{\SRF}{\mathsf{SRF}}

\newcommand{\bE}{\boldsymbol{E}}
\newcommand{\bF}{\boldsymbol{F}}
\newcommand{\bR}{\boldsymbol{R}}

\newcommand{\bC}{\boldsymbol{C}}
\newcommand{\bPsi}{\boldsymbol{\psi}}
\renewcommand{\bf}{\boldsymbol{f}}
\newcommand{\br}{\boldsymbol{r}}
\newcommand{\be}{\boldsymbol{e}}
\newcommand{\beps}{\Bold{\epsilon}}

\newcommand{\cB}{\mathcal{B}}
\newcommand{\C}{\mathcal{C}}
\newcommand{\D}{\mathcal{D}}
\newcommand{\cF}{\mathcal{F}}

\newcommand{\cS}{\mathcal{S}}

\newcommand{\cZ}{\mathcal{Z}}

\newcommand{\ERR}{\mathcal{E}^1_\Lambda}
\newcommand{\ERRd}{\mathcal{E}^2_{\Lm}}
\newcommand{\HYB}{\mathcal{H}^1_{\Li\Lu,\beps}}
\newcommand{\HYBd}{\mathcal{H}^2_{\Lmi\Lu,\beps}}
\newcommand{\HYBp}{\cB_{\Le\Lv,\bfR}^{1}}
\newcommand{\HYBpd}{\cB_{\Lme\Lv,\bfR}^{2}}

\newcommand{\SR}{S_{\bR}}

\newcommand{\SRinf}{S_{\bR}^{\infty}}
\newcommand{\SE}{S_{\bE}}
\newcommand{\SEh}{S_{\bE}^{h}}
\newcommand{\SEinf}{S_{\bE}^{\infty}}

\newcommand{\tm}{\bar{t}} 
\newcommand{\tmi}{\bar{t}_{i}}
\newcommand{\tme}{\bar{t}_{e}}
\newcommand{\SL}{\mathbb{S}_{\Bold{\lambda}}^{\ell}}

\newcommand{\Bold}[1]{\boldsymbol{#1}}

\newcommand{\fail}{f\!ail}

\DeclareMathOperator{\remop}{rem}
\newcommand{\rem}{\, \remop \, }

\DeclareMathOperator{\Val}{val}
\DeclareMathOperator{\vr}{v}
\DeclareMathOperator{\Ev}{Ev}
\newcommand{\Evi}{\Ev^{\infty}}

\newcommand{\dist}{d}

\newcommand{\commentaire}[1]{}

\theoremstyle{plain}
\newtheorem{theorem}{Theorem}[section]
\newtheorem{prop}[theorem]{Proposition}
\newtheorem{lemma}[theorem]{Lemma}
\newtheorem{coro}[theorem]{Corollary}

\theoremstyle{definition}
\newtheorem{exam}[theorem]{Example}
\newtheorem{defi}[theorem]{Definition}
\newtheorem{prob}[theorem]{Problem}
\newtheorem{constraint}[theorem]{Constraint}

\theoremstyle{remark}
\newtheorem{remark}[theorem]{Remark}

\newcommand{\htu}{{t}_u}
\newcommand{\hti}{{t}_i}
\newcommand{\hte}{{t}_e}
\newcommand{\htv}{t_v}
\newcommand{\Lu}{\Lambda_{u}}
\newcommand{\Li}{\Lambda_{i}}
\newcommand{\Lv}{\Lambda_{v}}
\newcommand{\Le}{\Lambda_{e}}
\newcommand{\Lm}{\Lambda_{m}}
\newcommand{\Lmi}{\Lambda_{m,i}}

\newcommand{\Lme}{\Lambda_{m,e}}
\newcommand{\xii}{\xi_{i}}
\newcommand{\xiu}{\xi_{u}}

\newcommand{\ie}{\textit{i.e.\ }}

\newcommand{\wrt}{\textit{w.r.t.\ }}

\newcommand{\lf}{\left\lfloor}   
\newcommand{\rf}{\right\rfloor}

\allowdisplaybreaks

\journal{ }

\begin{document}
	
	\begin{frontmatter}

		\title{Simultaneous Rational Function Codes: \\Improved Analysis Beyond Half the Minimum Distance \\with Multiplicities and Poles}

		\author{Matteo Abbondati,
			Eleonora Guerrini,
			Romain Lebreton}

		\affiliation{organization={LIRMM - University of Montpellier},
			addressline={161, Rue Ada}, 
			city={Montpellier},
			postcode={34095}, 
			country={FRANCE}}

		\begin{abstract}
			
			In this paper, we extend the work of~\cite{abbondati2024decoding} on decoding simultaneous rational function codes by addressing two important scenarios: multiplicities and poles (zeros of denominators). 
			
			First, we generalize previous results to rational codes with multiplicities by
			considering evaluations with multi-precision. Then,
			using the hybrid model from~\cite{guerrini2023simultaneous}, we extend our approach to vectors of
			rational functions that may present poles. 
			
			Our contributions include: a rigorous analysis
			of the decoding algorithm's failure probability that generalizes and improves several previous
			results, an extension to a hybrid model handling situations where not
			all errors can be assumed random, and a new 
 improved analysis in the more general context handling poles within multiplicities. The theoretical results provide a
			comprehensive probabilistic analysis of reconstruction failure in these
			more complex scenarios, advancing the state of the art in error
			correction for rational function codes.

		\end{abstract}

	\end{frontmatter}
	
	\section{Introduction}
	
			An efficient approach to solving linear systems in distributed computation involves reconstructing  
	a vector of fractions $\left(\frac{f_1}{g},\ldots,\frac{f_\ell}{g}\right)$, all sharing the same denominator,  
	from its modular reductions with respect to $n$ pairwise coprime elements. In this framework, a network  
	is structured around a central node, which selects a sequence of relatively prime elements  
	$\left(m_j\right)_{1 \le j \le n}$ and delegates the system solving process to the network.  
	Each node $j$ computes the solution modulo $m_j$ and transmits the reduced solution vector  
	$\left(f_i/g\ \bmod m_j\right)_{1 \le i \le \ell}$ back to the central node. The central node then  
	reconstructs the original vector through an interpolation step, formulated as a simultaneous  
	rational reconstruction problem. In the case of polynomial systems, this approach is known  
	as evaluation-interpolation~\cite{kaltofen2017early,guerrini2019polynomial}, whereas for integer  
	systems, it corresponds to modular reduction followed by reconstruction via the Chinese Remainder  
	Theorem~\cite{cabay1971exact}.
	
	This work and~\cite{abbondati2025simultaneous} are companion papers, addressing the polynomial and the integer contexts respectively.
	In view of the symmetry between these two cases, the layout of this paper reflects the structure of~\cite{abbondati2025simultaneous}.

	\paragraph{Context of this paper}
	During data reconstruction, the central node may receive incorrect reductions due to computational errors,  
	faulty or untrusted nodes, or network noise. 
	In view of this issue, it is of great help to look at the theory of error correcting codes and, by integrating decoding algorithms in the above framework, obtain fault tolerant linear system solving methods. 
	Viewing the modular
	reductions as coordinates of an error correcting code  enables us to reconstruct the correct solution as long as the number
	of erroneous reductions is below a certain value, corresponding to the unique decoding radius of the
	code. In presence of more errors, there exist two
	possible approaches in coding theory to correct beyond the unique decoding radius;
	either decoding algorithms which return a list of all codewords within a certain distance of the received word (list decoding) or, by interleaving techniques, obtain positive decoding results under probabilistic assumptions on random errors corrupting $\ell$ code-words on the shared positions. 
	In this paper, we focus on interleaving techniques as they fit in the simultaneous reconstruction problem. 
	Note that, a decoding algorithm working under this latter approach must inevitably fail for some instances, as
	beyond unique decoding radius there can be many codewords around a given instance. Here the failure
	probability is intended as the proportion of received words, within a given distance from
	the codeword $\bf/g$, for which the reconstruction fails. 
	
	Thus, focusing on the polynomial context, we will consider the simultaneous rational reconstruction problem over a finite field $\F_q$.
	In particular, we study the reconstruction with multiplicities, \ie 
	 with $m_j = (x - \alpha_j)^{\lambda_j}$ for a sequence of distinct evaluation points $\alpha_1,\ldots,\alpha_n\in\F_q$ and relative multiplicities $\lambda_1,\ldots,\lambda_n>0$.
	 The evaluation points for which the modular reductions are not defined (zeros of the denominator $g$) are referred to as \textit{poles}.
	One advantage of considering reductions with multiplicities is that
	solving a linear system modulo $(x - \alpha)^\lambda$ is asymptotically faster than solving
	it modulo $\alpha_1,\dots,\alpha_\lambda$ (see \cite{moenck_approximate_1979,dixon_exact_1982,storjohann_shifted_2005} or \cite[Chapter 3]{lebreton_contributions_2012} for a survey).

    In order to generalize our first results (Theorems~\ref{thm:main1} and~\ref{thm:main2}) on the decoding with multiplicities to a context including poles (Theorems~\ref{thm:main1poles} and~\ref{thm:main2poles}), we need to consider (as in~\cite{guerrini2023simultaneous}) a hybrid error model consisting of random evaluation errors and fixed valuation errors. 
	In Section~\ref{sec:Hybrid Decoding} we present this error model for more general sets of errors. We note that this more general version has been independently introduced in~\cite{brakensiek2025unique} for IRS codes and folded Reed-Solomon codes decoding, as the \textit{semi-adversarial error model}. There, the authors obtained a weaker failure probability bound (depending linearly instead of exponentially on the error size) but, for a given total amount of errors (random and fixed), the proportion of adversarial (fixed) errors they manage to correct, though asymptotically equivalent, is higher compared to our results.


	\paragraph{Previous results}
	
	A long series of papers can be found in the literature where
	evaluation-interpolation is used for linear systems solving, as
	\cite{mcclellan1977exact,villard1997study,monagan2004maximal,olesh2007vector,rosenkilde2016algorithms}.
	In the  polynomial case, the codes
	used for the recovery of a vector of polynomials from partially erroneous
	evaluations are Interleaved Reed-Solomon codes (IRS), whose best known analysis
	of the decoding failure probability is provided
	in~\cite{schmidt2009collaborative} and then generalized to the rational function
	case in~\cite{abbondati2024decoding}.
	The integer counterpart of IRS codes are to the so-called Interleaved Chinese remainder codes (ICR),
	for which a first heuristic analysis of the decoding failure probability  was provided
	in~\cite{li2013decoding} and made rigorous in~\cite{abbondati2023probabilistic}. 
	While there have been various studies on the rational function
	case~\cite{kaltofen2017early,zappatore2020simultaneous,guerrini2023simultaneous}, the rational
	number context had not been investigated until~\cite{abbondati2024decoding}.

	The extensive literature addressing
	these problems both in the polynomial
	\cite{mcclellan1977exact,boyer2014numerical,guerrini2019polynomial,kaltofen2020hermite,guerrini2023simultaneous}
	and the integer \cite{cabay1971exact,lipson1971chinese,abbondati2023probabilistic} contexts rarely
	shows unified methods, and the techniques used are very specific to the case studied. In~\cite{abbondati2024decoding} we provided (for the pole-free case and in absence of multiplicities) a unified analysis technique (inspired from~\cite{schmidt2009collaborative}) for both contexts, proving that it is possible to recover the correct solution vector for almost all instances.
	The results of~\cite{abbondati2024decoding} have been generalized to include both multiplicities and poles for the rational number case in~\cite{abbondati2025simultaneous}.
	The current paper, while keeping the same analysis of~\cite{abbondati2025simultaneous}, extends its results to the rational function context.

	\paragraph{Contributions of this paper}
	The approach of this paper generalizes and matches (or even improves) several
	previous results in different ways.

	The reconstruction of a vector of rational functions (in the general case considered here of multiplicities and poles) has already been addressed in~\cite{guerrini2023simultaneous},
	where the argument used to derive the failure probability bound is based (as in other related sources~\cite{bleichenbacher2003decoding,brakensiek2025unique}) on Schwartz-Zippel Lemma, resulting in a linearly dependent bound as a function of the error size  (\ie $\P_{\fail} \leq e/q$).
	
	The approach considered here (and in~\cite{abbondati2024decoding} for the special case of \textit{rational evaluation codes} with no multiplicities nor poles), derived from techniques used in the analysis of syndrome decoding for interleaved Reed-Solomon codes~\cite{schmidt2009collaborative}, allows to improve the probability bound to an exponentially decreasing function of the error size.
	
	Furthermore, the approach presented here does not require the \textit{multiplicity balancing}~\cite[$\S$3.1]{guerrini2023simultaneous}, which reflected a dependence of the error correction capacity on the distribution of error multiplicities among the $\ell$ codewords, thus improving the bound~\cite[Theorem 3.4]{guerrini2023simultaneous} on the maximal number of errors we can correct. 
	Nevertheless, we remark that in~\cite{brakensiek2025unique} (in absence of multiplicities and for the polynomial case), though the probability bound is still linear, 
	the trade-off between adversarial errors and distance is stronger than ours (Theorems~\ref{thm:main1hyb} and~\ref{thm:main2hyb}).

	In short the main results presented are the following:
	\begin{itemize}
		\item A detailed analysis of the failure probability of the decoding Algorithm~\ref{algoSRF}, that generalizes previous results in \cite{abbondati2024decoding} to the multiplicity case: see Theorem~\ref{thm:main1} and Theorem~\ref{thm:main2}.
		\item The extension of the analysis to a hybrid model including random and non-random errors, addressing situations where not all errors can be assumed random: see Theorem~\ref{thm:main1hyb}, Theorem~\ref{thm:main2hyb}. 
		\item The merging of the hybrid model with our decoding approach, to handle poles  multiplicities, and relative decoding failure analysis: see Theorem~\ref{thm:main1poles} and Theorem~\ref{thm:main2poles}, improving several previous results~\cite{kaltofen2020hermite,guerrini2023simultaneous}. 
	\end{itemize}

\subsection{Notations and preliminary definitions}
\label{sec:NotationDefinitions}
We will denote vectors with bold letters $\bf,\br,\Bold{c},\ldots$.
For $m\in\F_q[x]$ we will denote its degree with $\partial(m)$ and for $\bf \in \F_q[x]^{\ell}$ we define its degree as the max $\partial(\bf) \coloneq \max_i\{\partial(f_i)\}$. $\F_q[x]/m$ will denote the quotient ring modulo the ideal $(m)$, while $\cZ(m)$ will denote the set of roots
of $m$ in $\F_q$. 
Given an indexed family of rings $\left\{A_j\right\}_{1 \le j \le n}$, we let $\prod_{j=1}^n A_j$ be their
Cartesian product.
Given a vector of modular reductions
$\br\in\prod_{j=1}^n\F_q[x]/(x - \alpha_j)^{\lambda_j}$ the corresponding capital letter $R$ denote its
unique interpolant constructed via the Chinese remainder theorem modulo
$M\coloneq\prod_{j=1}^n(x - \alpha_j)^{\lambda_j}$.

We let $\Val_{\alpha}:\F_q[x] \longrightarrow
\N\cup\{\infty\}$ be the valuation function over $\F_q[x]$ with respect to the evaluation point $\alpha\in\F_q$, whose output is the
highest power of $x - \alpha$ dividing the input, where we set by convention its
value to be $\infty$ when the input is $0$.

Dealing with a fixed sequence of precisions $\lambda_1,\ldots,\lambda_n$, we truncate the valuation
function considering $\nu_{\alpha_j}(m) \coloneq\min\{\Val_{\alpha_j}(m),\lambda_j\}$, so that $\nu_{\alpha_j}(a) = \nu_{\alpha_j}(b)$ when $a = b \bmod (x - \alpha_j)^{\lambda_j}$.

When computing the valuation of a vector we set
$\nu\left(\bf\right)\coloneq\min_i\{\nu\left(f_i\right)\}$. Given the sequence of
multiplicities $\lambda_1,\ldots,\lambda_n>0$, we define the parameter
$L\coloneq\sum_{j=1}^n\lambda_j$. For us, all the vectors of fractions $\bf/g$ sharing the same
denominator will always be reduced, i.e. they satisfy $\gcd\left(\gcd(\bf),g\right) = 1$.

\paragraph{Simultaneous rational function reconstruction with errors (SRFRwE)}

To quantify errors and to establish the correction capacity of the code we are
going to use, we need a notion of distance between words. In a context with
multiplicities where the coordinates are modular reductions relative to moduli
specified by different precisions $\lambda_1,\ldots,\lambda_n$, it is classical
to consider (see for example~\cite{kaltofen2020hermite,guerrini2023simultaneous}) a minimal error
index distance in which each modular reduction is regarded as a truncated
development, and the whole tail starting from the first error index in such
development is considered erroneous. Thus,
we are going to consider the following definition:

\begin{defi}[Distance]
	\label{def:PolyDistance}
	Let \mbox{$\bR^1, \bR^2 \in (\prod_{j=1}^n\F_q[x]/(x - \alpha_j)^{\lambda_j})^{\ell}$} be two
	\mbox{$\ell\times n$} matrices, where each column $\br_j^1,\br_j^2$ belongs to $\left(\F_q[x]/(x - \alpha_j)^{\lambda_j}\right)^\ell$. We
	define their error support as  
	$\xi_{\bR^1,\bR^2} \coloneq \{ j: \br_j^1 \ne \br_j^2\}$ and their error
	locator polynomial as the product $\Lambda_{\bR^1,\bR^2}\coloneq\prod_{j \in \xi_{\bR^1, \bR^2}}
	(x - \alpha_j)^{\lambda_j - \mu_j}$, where $\mu_j\coloneq\nu_{\alpha_j}\left(\br_j^1  -
	\br_j^2\right)$ represents the minimal error index for the development around the
	evaluation point $\alpha_j$. The distance between $\bR^1$ and $\bR^2$ is defined as
	$d(\bR^1,\bR^2)\coloneq \partial(\Lambda_{\bR^1,\bR^2})$.
\end{defi}

The problem of simultaneous rational function reconstruction with errors is then given by:

\begin{prob}[SRFRwE]
	\label{SRFRwE_Problem}
	Given $\ell >0$, $n$ distinct evaluation points $\alpha_1,\ldots, \alpha_n\in\F_q$ with associated multiplicities
	$\lambda_1,\ldots,\lambda_n$, a received matrix $\bR \in
	(\prod_{j=1}^n\F_q[x]/(x - \alpha_j)^{\lambda_j})^{\ell}$, an error parameter $t$ and two degree bounds $d_f, d_g$ such
	that $d_f + d_g \leq L + 1$, find a reduced vector of rational functions $\left(f_1/g,\ldots,
	f_\ell/g\right)\in\F_q(x)^{\ell}$ such that
	\begin{enumerate}
		\item $d\left(\bigl(f_i/g \ \bmod (x - \alpha_j)^{\lambda_j}\bigr)_{i,j},\bR\right)\leq t$,
		\item $\partial(\bf)<d_f$, $\partial(g)<d_g$ and $\gcd(g,M)=1$.
	\end{enumerate}
\end{prob}

In the above we have that $\gcd(g,M)=1$ so that the
reductions $f_i/g \bmod (x - \alpha_j)^{\lambda_j}$ are well-defined. We are going to drop this hypothesis in
Section~\ref{Sec:poles}, when solving a more general version of
the SRFRwE problem, allowing for the presence of poles.

This problem can be reduced to the simultaneous error correction of $\ell$ code words
(sharing the same denominator) for the multiplicity version of rational function codes. Without multiplicities (i.e. when $M$ is square-free)
this code is the natural rational extension of Reed-Solomon codes \cite{reed1960polynomial}, and can be referred to as rational function codes,
extensively studied in~\cite{abbondati2024decoding}. It seems these rational codes were part of the
folklore; to the best of our knowledge, they were formally introduced in the language of coding
theory by Pernet in \cite[$\S$ 2.5.2]{pernet2014high}.

The condition $d_f + d_g \leq L + 1$ guarantees an injective encoding, whose proof will be given in Proposition~\ref{prop:InjectiveEncodingpoles}
when considering the multi-precision encoding (see Definition~\ref{def:MultiprecisionEncoding}) first introduced in~\cite{guerrini2023simultaneous}, which is a
generalization of our current encoding in presence of poles.

This paper is concerned with error correction beyond guaranteed
uniqueness, this means that the solution to the problem will not always be
unique. In this rare case, our decoding algorithm returns a decoding failure.
We analyze the probability of failure in detail.

\medskip

The paper is structured as follows: In Section~\ref{Sec:SRF} we introduce the
simultaneous rational function codes whose decoding solves
Problem~\ref{SRFRwE_Problem} as well as the corresponding decoding
Algorithm~\ref{algoSRF}. We study the failure probability of our decoding
algorithm for error parameters larger than the unique decoding radius of the
code. We note that this analysis generalizes the results
of~\cite{abbondati2024decoding} to the multiplicity case, it thus follows the
same broad lines except for some technical details (see
Lemma~\ref{lm:sumOverDivisorsPOLY}).

In Section~\ref{sec:Hybrid Decoding}, we adapt our analysis technique to the hybrid distribution
model of~\cite{guerrini2023simultaneous} in which not all errors are supposed to be random, but some
of them are fixed, either because of specific error patterns introduced by malicious entities or
because of specific faults of the network nodes.

Then, in Section~\ref{Sec:poles}, by considering the multi-precision encoding
of~\cite{guerrini2023simultaneous}, we apply the hybrid approach to generalize our analysis to the
case of reductions with multiplicities and poles, \ie we drop the hypothesis
$\gcd\left(g, M\right) = 1$.

	\section{Simultaneous Multiplicity Rational Function Codes}\label{Sec:SRF}

	We reduce Problem~\ref{SRFRwE_Problem} to the decoding of simultaneous rational function
	codes with multiplicities, defined as follows:
	\begin{defi}
		\label{def:SRF}
		Given $n$ distinct evaluation points $\alpha_1,\ldots,\alpha_n\in\F_q$ with relative
		multiplicities $\lambda_1,\ldots,\lambda_n\in\Zpos$, let $M(x)\coloneq\prod_{j=1}^n
		(x-\alpha_j)^{\lambda_j}\in\F_q[x]$ and $L \coloneq \partial\left(M\right) =
		\sum_{j=1}^{n}\lambda_j$, two degree bounds $d_f,d_g\in\Z_{>0}$ such that $d_f+d_g\leq L + 1$
		and a parameter $\ell>0$, we define the \textit{simultaneous multiplicity rational function code} as the set of matrices 
		
		\begin{equation*}
			\SRF_{\ell}(M;d_f,d_g)\coloneq
			\left\{
			\begin{pmatrix}
				\frac{f_i}{g}\ \bmod (x - \alpha_j)^{\lambda_j}
			\end{pmatrix}_{\substack{1 \leq i \leq \ell \\ 1 \leq j \leq n}}:
			\begin{array}{c}
				\partial(\bf)<d_f, \partial(g)<d_g,\\
				\gcd(f_1,\ldots,f_\ell,g)=1\\
				\gcd(M,g)=1
			\end{array}
			\right\}.
		\end{equation*}
	We will refer to SRF codes for short if the parameters are implicit.
	\end{defi}
	
	Note that when $L = n$, \ie $\lambda_j = 1$ for every $j=1,\ldots,n$, we obtain the simultaneous rational function codes first introduced in~\cite[$\S$ 2.3]{pernet2014high} and extensively studied in~\cite[$\S$ 4]{abbondati2024decoding}. If furthermore $d_g = 1$ we obtain the interleaving of Reed-Solomon codes.
	
	In the next section, we will see that the common denominator property is necessary to be able to take advantage in the 
	key equations of the fact that the $\ell$ rational function codewords forming the lines of the matrices in $\SRF_{\ell}(M; d_f, d_g)$ share the same error supports.

	The condition $\gcd(f_1,\ldots,f_\ell,g)=1$, which is going to be used in the proof of
	Lemma~\ref{BaseLemmaRN}, reflects that the solution vector we seek to reconstruct is a reduced
	vector of rational functions.
	\begin{remark}
		A bounded distance decoding algorithm for the above code which is able to correct errors up to a
		distance $t$, can be used to solve Problem~\ref{SRFRwE_Problem} with error parameter $t$.
        This justifies the denomination \textit{simultaneous} for the above code.
	\end{remark}

	\subsection{Decoding algorithm}

	Let $\bR \coloneq (r_{i,j})_{\substack{1 \leq i \leq \ell \\ 1 \leq j \leq n}}$ be the received
	matrix. For any code word $\bC \in SRF_{\ell}(M;d_f,d_g)$, we can write $\bR=\bC+\bE$ for some error
	matrix $\bE$. We can associate an interpolation polynomial to every row, which we write
	$R_i=C_i+E_i$.
	We know that $\Lambda f_i = \Lambda g R_i \bmod M(x)$ holds for any $1 \le i \le
	\ell$~\cite{pernet2014high}. Making the substitutions $\varphi \leftarrow \Lambda g,\psi_i
	\leftarrow \Lambda f_i$ we linearize the previous equations, obtaining the \emph{key equations} 
	\begin{equation*}
		\psi_i = \varphi R_i \bmod M(x) \text{ for } i=1,\ldots,\ell
	\end{equation*}
	which are $\F_q$-linear. In particular if $\partial(\Lambda) = d\left(\bR,\bC\right) \leq t$ for some distance parameter
	$t$ (input of Problem~\ref{SRFRwE_Problem}), the solution vector $v_{\bC}\coloneq(\Lambda g,\Lambda f_1,\ldots,\Lambda f_\ell)$ belongs to
	the $\F_q$-linear subspace
	\begin{equation*}
		\SR \coloneq\left\{
		(\varphi,\psi_1,\ldots,\psi_\ell)\in\F_q[x]^{\ell +1}:
		\begin{array}{c}
			\psi_i = \varphi R_i\bmod M(x) \\
			\partial(\varphi)<d_g+t, \partial(\Bold{\psi})<d_f+t
		\end{array}
		\right\}.
	\end{equation*}

%

	Taking inspiration from~\cite[Algorithm 2.1]{kaltofen2017early} the decoding algorithm for SRF codes with multiplicities is based on the computation of a minimal degree element in the set of solutions $\SR$: 
	
	\begin{algorithm}[H]
		\caption{$\SRF_\ell$ codes decoder.}
		\label{algoSRF}
		\SetAlgoLined \SetKw{KwBy}{par} \KwIn{$\SRF_\ell(M;d_f,d_g)$, received word $\bR:=(\br_j)_{1\le j\le n}$, distance
			bound $t$} \KwOut{A reduced vector of fractions $\bPsi'/\varphi'$ s.t. $d\left(\bigl(\psi_i'/\varphi' \ \bmod (x - \alpha_j)^{\lambda_j}\bigr)_{i,j},\bR\right) \leq t$ or ``decoding failure''}
		
		\vspace{5pt}
		
		Compute $\Bold{0}\ne v_s \coloneq \left(\varphi,\psi_1,\ldots,\psi_{\ell}\right)\in\SR$, s.t. $\max\{\partial(\varphi),\partial(\Bold{\psi})\}$ is minimal.\\
		Let $\eta \coloneq \gcd(\varphi,\psi_{1},\dots,\psi_{\ell})$, $\varphi'\coloneq\varphi/\eta$ and
		$\forall i, \ \psi'_i\coloneq\psi_i/\eta$\\ 
		\If{$\partial(\eta) \le t$, $\partial(\varphi') < d_g, \partial(\Bold{\psi}')< d_f$}
		{\textbf{return} $(\psi_1'/\varphi', \dots, \psi_\ell'/\varphi')$}
		\lElse{ \textbf{return} "decoding failure"}
	\end{algorithm}

\begin{lemma}
\label{lm:algofailureconditionPOLY}
If the decoding algorithm fails then $v_s \notin v_{\bC} \F_q[x]$.
\end{lemma}
	
\begin{proof}
	We prove it by contrapositive. If $v_s \in v_{\bC} \F_q[x]$, \ie $v_s = p v_{\bC}$ for some $p\in\F_q[x]$. Given that $v_s$ has minimal degree among the solutions, and that $v_{\bC}\in\SR$, we conclude that $p\in\F_q$. Therefore, since the vector of rational functions $\bf/g$ is assumed to be reduced, we conclude that $\eta = \Lambda$, thus $(\varphi', \Bold{\psi'}) = (pg,p\bf)$ and $\partial(\varphi') < d_g$, $\partial(\Bold{\psi'})< d_f$ therefore the decoding algorithm succeeds.
\end{proof}
	
\subsection{Minimal distance}

	The distance $d(\C)\coloneq\min_{c_1 \neq c_2\in \C} d(c_1,c_2)$ of a code $\C$ plays an important
role in coding theory to assess the amount of data one can correct. A classic result states that when the distance between the code $\C$ and the input received word $\bR$, defined as $\min\{d(\bC,\bR): \bC\in\C\}$, is below half the minimal distance of $\C$, then the decoding algorithm is guaranteed to succeed, while there is no guarantee on the decoding success beyond this quantity.
We start by proving a lower bound for the minimal distance of SRF codes:
		\begin{lemma}
	\label{lm:minDistSRFcode}
	We have $d(\SRF_{\ell}(M;d_f,d_g)) > L - d_f - d_g + 1 $.
\end{lemma}
\begin{proof}
	Let  $\bC_1= \left(f_i/g \ \bmod (x - \alpha_j)^{\lambda_j}\right)_{i,j}$ and $\bC_2=
	\left(f_i'/g'\ \bmod (x - \alpha_j)^{\lambda_j}\right)_{i,j}$ be two distinct code words. Setting
	$Y\coloneq\prod_{j=1}^n (x - \alpha_j)^{\mu_j}$, with $\mu_j =
	\nu_{\alpha_j}\left(\bf/g - \bf' /g'\right)$. Thanks to
	$\gcd\left(Y,g\right) = \gcd\left(Y,g'\right)= 1$, we have that  
	$Y| (\bf g' - \bf' g)$, with $\bf g' - \bf' g \ne \Bold{0} \bmod M$ since $\bC_1 \ne \bC_2$. Given that  $\partial(\bf), \partial(\bf')< d_f$, and
	$\partial(g), \partial(g')< d_g$ we have $\partial(Y)< d_f + d_g - 1$. Using the relation $Y=M/ \Lambda_{\bC_1,\bC_2}$,
	we bound $d(\bC_1,\bC_2) = \partial(\Lambda_{\bC_1,\bC_2}) = \partial(M/Y) > L - d_f - d_g + 1$.
\end{proof}

\begin{remark}
	If $M$ is square-free, in~\cite[$\S$ 2.3.1]{pernet2014high} the author showed, by constructing two distinct codewords at a given distance, that $d(\SRF_{\ell}(M;d_f,d_g)) = L - d_f - d_g + 2$. To adapt the same proof in the context of multiplicities we need to drop the hypothesis $\gcd(g,M) = 1$, thus we will prove it in Section~\ref{Sec:poles} when dealing with poles (see Lemma~\ref{lm: MinDistSRFpoles}).  
\end{remark}

	\subsection{Error Models}
\label{subsect:ErrModelRatFunc}

Decoding Algorithm~\ref{algoSRF} must fail on some instances when the distance parameter $t$ exceeds the
maximum distance for which the uniqueness of the solution of Problem~\ref{SRFRwE_Problem} is
guaranteed.

We analyze the failure probability of the algorithm under two different classical error models in
Coding Theory, already considered in previous papers
\cite{schmidt2009collaborative,abbondati2023probabilistic,abbondati2024decoding,abbondati2025simultaneous}, specifying two
possible distributions of the random received word $\bR$.

\paragraph{Error Model 1}
In this error model, we fix an error locator $\Lambda$ among the divisors of
$M$, then we let $\ERR$ be the set of error matrices $\bE$ whose columns $\be_j\in\left(\F_q[x]/(x - \alpha_j)^{\lambda_j}\right)^{\ell}$ satisfy:
\begin{enumerate}
	\item $\be_j = \boldsymbol{0} \text{ for all } j \text{ such that } \alpha_j\not\in\cZ(\Lambda)$,
	\item $\nu_{\alpha_j}(\be_j) = \lambda_j - \nu_{\alpha_j}(\Lambda) \text{ for all } j \text{ such that } \alpha_j\in\cZ(\Lambda)$.
\end{enumerate}
For any given code word $\boldsymbol{C}$ and error locator $\Lambda$, a random
received word $\bR$ around the central
code word $\bC$ is of the form $\bR = \bC + \bE$ for $\bE$ uniformly distributed
in $\ERR$.

We will need another point of view on the random error matrices $\bE$. For
$i\in\{1,\dots,\ell\}$, we denote $E_i \in \F_q[x]/M$ the interpolant of the
$i$-th row of $\bE$. By definition of the minimal error index $\mu_j$, letting $Y
\coloneq M/\Lambda = \prod_{j=1}^n (x - \alpha_j)^{\mu_j}$, we have that $Y | E_i$ for
every index $i= 1,\ldots,\ell$. We define the modular reduction of the quotients $E_i'\coloneq
E_i/Y \in \F_q[x]/\Lambda$.

Since $\mu_j = \nu_{\alpha_j}(\bE) = \min_i\{\nu_{\alpha_j}(E_i)\}$, we see that $Y =
\gcd(E_1,\ldots,E_{\ell},M)$, and that the random vector $(E_i')_{1\le i \le \ell}$ is uniformly
distributed in the sample space 
\begin{equation}
	\Omega_{\Lambda,\ell} \coloneq
	\{ (F_i)_{1\le i \le \ell} \in (\F_q[x]/\Lambda)^\ell : \gcd(F_1,\dots,F_\ell,\Lambda) = 1\}.
\end{equation}
As we will need a more general version of $\Omega_{\Lambda,\ell}$ (for example in the proof of
Lemma~\ref{lm:probagpe}), we state the following:
\begin{lemma}
	\label{lm:EulerFormula}
	Fix a polynomial $\Lambda\in\F_q[x]$ having all its roots in $\F_q$. For any of its divisors $\eta = \prod_{\alpha\in\cZ(\Lambda)}(x - \alpha_j)^{\eta_j}$, letting 
	$$\Omega_{\Lambda,\eta,\ell} := \left\{ (F_i)_{1\le i \le \ell} \in (\F_q[x]/\Lambda)^\ell :
	\gcd(F_1,\dots,F_\ell,\Lambda) = \eta \right\} $$
	 we have
	\begin{equation*}
		\#\Omega_{\Lambda,\eta,\ell}
		=
		q^{\ell \left(\partial(\Lambda/\eta)\right)}
		\prod_{\alpha\in\cZ(\Lambda/\eta)}\left(1 - \frac{1}{q^{\ell}}\right)
	\end{equation*}
\end{lemma}
\begin{proof}
    We note that the function
    \begin{align*}
        \Omega_{\Lambda,\eta,\ell}&\longrightarrow\Omega_{\Lambda/\eta,\ell}\\
        (F_i)_{1\le i \le \ell}&\longmapsto\frac{1}{\gcd(F_1,\ldots,F_{\ell},\Lambda)}(F_i)_{1\le i \le \ell}
    \end{align*}
    is a bijection, therefore 
    $$
    \#\Omega_{\Lambda,\eta,\ell} = \#\Omega_{\Lambda/\eta,\ell}=\#\left\{\bF\coloneq (F_i)_{1\le i \le \ell}\in \left(\F_q[x]/\left(\Lambda/\eta\right)\right)^\ell : \forall j = 1,\ldots,n, \ \ \nu_{\alpha_j}(\bF) = 0\right\}.
    $$

	By counting the vectorial coefficients of the development of $\bF$ in the neighborhood of each point $\alpha_j$, we
	obtain the cardinality of the above set as
	\begin{equation*}
		\prod_{\eta_j<\nu_{\alpha_j}(\Lambda)}
		q^{\ell(\nu_{\alpha_j}(\Lambda) - \eta_j - 1)}
		\left(
		q^{\ell} - 1
		\right)
		=
		q^{\ell \left(\partial(\Lambda/\eta)\right)}
		\prod_{\alpha\in\cZ(\Lambda/\eta)}\left(1 - \frac{1}{q^{\ell}}\right). \qedhere
	\end{equation*}
\end{proof}

\paragraph{Error Model 2}
In this error model we fix a maximal error locator $\Lm$ among the divisors of $M$, then we let $\ERRd$ be the set of error matrices $\bE$ whose columns satisfy:
\begin{enumerate}
	\item $\be_j = \boldsymbol{0} \text{ for all } j \text{ such that } \alpha_j\not\in\cZ(\Lm)$,
	\item $\nu_{\alpha_j}(\be_j) \ge \lambda_j - \nu_{\alpha_j}(\Lm) \text{ for all } j \text{ such that } \alpha_j\in\cZ(\Lm)$.
\end{enumerate}

We notice that in the error model $\ERRd$, the actual error locator $\Lambda$ could be a divisor of
$\Lm$. For a code word $\boldsymbol{C}$ and a maximal error locator $\Lm$, a random
 received word $\bR$ around the central
code word $\bC$ is of the form $\bR = \bC + \bE$ for $\bE$ uniformly distributed
in $\ERRd$.

\subsection{Our Results}
\label{SubSect.Results}
In this section we present our contributions to the analysis of the decoding failure depending on
the given parameters. The error models previously defined will play a role in the analysis but not in
the choice of parameters. We  define a common framework for the algorithm parameters, while in 
Subsection~\ref{Sec.AnalysisRN} we will adapt the analysis of the failure probability to the two error
models specified above.
In what follows we set
	\begin{equation}
	\label{tmaxPoly}
	\tm\coloneq\frac{\ell}{\ell +1}\left(L - d_f - d_g + 1 \right)
\end{equation}

\begin{remark}
	Our setting allows decoding up to a distance $t\le \tm$ that, for $\ell > 1$, is greater than $\lf\frac{L - d_f - d_g + 1}{2}\rf$, thus in some cases it can allow to go beyond the unique decoding capability
	of $\SRF_{\ell}(M;d_f,d_g)$ codes.
    Indeed, even though from Lemma~\ref{lm: MinDistSRFpoles} we know that the quantity $\lf\frac{L - d_f - d_g + 1}{2}\rf$ is smaller than half of the minimal distance of the code, the decoding capability expressed by~\eqref{tmaxPoly} asymptotically reaches the quantity $L - d_f - d_g + 1$, which (if enlarging the code by allowing poles and under the hypothesis of Constraint~\ref{cst:SubsetSumThresholdPoles}) as we will see in Lemma~\ref{lm: MinDistSRFpoles}, is equal to the minimal distance of the code minus one, thus larger than half the minimal distance (unless we are in the trivial case in which the distance of the code is 2).
    In this section, for the code with no poles, codewords are fewer thus more far from each other, thus the minimal distance can in principle be strictly larger than in the case with poles. 
\end{remark}
When fixing the decoding bound $t$ close to  $\tm$, we are likely to correct beyond the unique
decoding radius, so we must deal with  decoding failure for some received word.

Here is our first result (whose proof will be given at the end of Subsection~\ref{RN_ERR1}) relative
to the failure probability of the decoding algorithm with respect to the error model $\ERR$.

\begin{theorem}
	\label{thm:main1}
	Decoding Algorithm~\ref{algoSRF} on input
	\begin{enumerate}
		\item distance parameter $t\le \tm$,
		\item a random received word $\bR$ uniformly distributed in  $\bC + \ERR$, for some code word $\bC \in \SRF_\ell(M;d_f,d_g)$
	and error locator $\Lambda$ such that $\partial(\Lambda) \leq t$,
	\end{enumerate}
	outputs the center  code word $\bC$ of the distribution
	with a probability of failure 
	\begin{equation*}
		\P_{\fail}
		\leq 
		\frac{q^{-(\ell +1)(\tm-t)}}{q-1}
		\prod_{\alpha\in \cZ(\Lambda)}
		\left(
		\frac{1 - 1/q^{\ell + \nu_{\alpha}(\Lambda)}}{1 - 	1/q^\ell}
		\right)
	\end{equation*} 
\end{theorem}

Here is our second result (whose proof will be given at the end of Subsection~\ref{RN_ERR2})
relative to the failure probability with respect to the error model $\ERRd$.

\begin{theorem}
	\label{thm:main2}
    Decoding Algorithm~\ref{algoSRF} on input
	\begin{enumerate}
		\item distance parameter $t\le \tm$,
		\item a random received word $\bR$ uniformly distributed in  $\bC + \ERRd$, for some code word $\bC \in \SRF_\ell(M;d_f,d_g)$
	and maximal error locator $\Lm$ such that $ \partial(\Lm) \leq t$,
	\end{enumerate}
	outputs the center  code word $\bC$ of the distribution
	with a probability of failure 
	\begin{equation*}
		\P_{\fail}
		\le 
		\frac{q^{-(\ell +1)(\tm-t)}}{q-1}
		\prod_{\alpha\in\cZ(\Lm)}
		\left(\frac{1 - 1/q^{\ell + 	\nu_\alpha(\Lm)}}{1 - 1/q^{\ell + 1}}\right).
	\end{equation*} 
\end{theorem}
We remark that both results reduce to~\cite[Theorem 24 and 25]{abbondati2024decoding} respectively,
when there are no multiplicities in the modular reductions of the code, \ie\ when $M$ is
square-free.

We note that in both theorems the product over the roots of the error locator is close to
one when $n<<q^{\ell}$; indeed we can prove the following lemma.

\begin{lemma}
	Given a divisor $\eta$ of $M$ and $f(\ell)>0$ any positive function of
	the parameter $\ell>0$, we have that 
	\begin{equation*}
		\prod_{\alpha\in \cZ(\eta)}
		\left(
		\frac{1 - 1 / q^{\ell + \nu_{\alpha}(\eta)}}{1 - 1 /q^{f(\ell)}}
		\right)
		\leq
		\frac{1}{1 - n / q^{f(\ell)}} \approx 1.
	\end{equation*}
    where the last approximation is in the range of parameters for which $n<<q^{f(\ell)}$.
\end{lemma}
\begin{proof}
	We start noticing that for each factor in the product we have
	\begin{equation*}
		\frac{1 - 1 / q^{\ell + \nu_{\alpha}(\eta)}}{1 - 1 /q^{f(\ell)}}
		\leq
		\frac{1}{1 - 1 / q^{f(\ell)}}
	\end{equation*}
	Furthermore $\prod_{\alpha\in \cZ(\eta)}(1 - 1/q^{f(\ell)}) \geq (1 - 1/q^{f(\ell)})^n \geq 1 - n/q^{f(\ell)}$, where the last inequality holds because for every $x\in\R$, by induction on $n\geq 1$, we have that
    \begin{equation*}
        (1 - x)^n \geq \left(1 - (n-1)x\right)(1 - x) = 1 - nx + (n - 1)x^2\geq 1 - nx
    \end{equation*}
\end{proof}

\begin{remark}
	We give a scenario which highlights how Theorem~\ref{thm:main2} can be used in practice. Assume
	that a code is fixed such that $L - d_f - d_g + 1 = 20$, so that with an interleaving parameter
	$\ell = 4$, one has $\tm = 16$. If one wishes to ensure that the failure probability is less
	than a target probability of $q^{-31}$, then Theorem~\ref{thm:main2} (where we approximate $\P_{fail}$ with $q^{-(\ell + 1)(\tm - t)}/q$) states that choosing the
	distance parameter of the decoder $t = 10$, ensures that for any random error uniformly
	distributed on a maximal error locator $\Lm$ such that $\partial(\Lm) \le t$, the failure
	probability is less than $q^{-31}$.	    
\end{remark}

\subsection{Analysis of the decoding failure probability}
\label{Sec.AnalysisRN}

For any $\bR$ uniformly distributed in $\D^{\ERR}_{\bC}$ (as in Theorem~\ref{thm:main1}), under the hypothesis $\partial(\Lambda)\le t$ we 
have that $v_{\bC} = (\Lambda g,\Lambda\bf)\in\SR\ne\emptyset$. We recall that the decoding Algorithm~\ref{algoSRF} computes a nonzero minimal degree solution $v_s\in\SR$.

\subsubsection{Decoding failure probability with respect to the first error model} 
\label{RN_ERR1}
If Algorithm~\ref{algoSRF} fails, then $v_s\notin v_{\bC}\F_q[x]$ (see
Lemma~\ref{lm:algofailureconditionPOLY}). Note that the converse is not necessarily true, for example if
there exists another close code word $\bC'\neq \bC$ with $d(\bC',\bR) \le t$ and if the decoding algorithm
outputs $v_s=v_{\bC'}$.
Nevertheless, we can upper bound the failure probability of the algorithm as
$\P_{\fail}\leq\P(S_{\bR}\not\subseteq  v_{\bC}\F_q[x])$. 

We introduce some notations: for
$C\in\Z$ we let 
\begin{equation*}
    \F_q[x]_{m,C} \coloneq \left\{ a \in \F_q[x]/m : \partial (a \rem m) \leq C \right\},    
\end{equation*}
where $a \rem m$ is the remainder of $a$ modulo $m$, that is the unique
representative of $a$ modulo $m$ whose degree is less than $\partial(m)$. Note
that this set has cardinality 
\begin{equation*}
    \#\F_q[x]_{m,C}
    =
    \begin{cases}
        1 & \text{if }\ C+1\le 0\\
        q^{C+1} & \text{if }\ 0 < C+1 \leq \partial(m)\\
        q^{\partial(m)} & \text{if }\ C + 1 > \partial(m).
    \end{cases}
\end{equation*}
We let $\SE$ be the set 
\begin{equation*}
    \SE\coloneq \left\{ \varphi\in\F_q[x]/\Lambda : \forall i, \ g\varphi E'_i \in
\F_q[x]_{\Lambda,B + \partial(\Lambda)} \right\}    
\end{equation*}
for $B\coloneq d_f + d_g + t - L - 2$.

We need a new constraint to prove the following lemma.
\begin{constraint}
	\label{c_3}
	Algorithm~\ref{algoSRF} parameters satisfy $B < 0$.
\end{constraint}

\begin{lemma}
	\label{BaseLemmaRN}
	If Constraint~\ref{c_3} is satisfied,  
	$\SE=\{0\} \Rightarrow S_{\bR}\subseteq  v_{\bC}\F_q[x]$.
\end{lemma}
\begin{proof}
	Let $(\varphi,\psi_1,\ldots,\psi_\ell)\in \SR$. We know that for all $1 \le i \le \ell$,
	$g\varphi E_i  = g\varphi\left(R_i-\frac{f_i}{g}\right)  = g\psi_i-f_i\varphi \bmod M.$ Since
	$Y|E_i$ and $Y|M$, thanks to the above, we have that $Y|(g\psi_i-f_i\varphi)$, and we define the
	polynomial $\psi'_i \coloneq \frac{g\psi_i-f_i\varphi}{Y}$. Dividing the above modular equation by $Y$ we
	obtain $g\varphi E'_i = \psi'_i \bmod \Lambda$. Therefore,
	\begin{equation*}
		\partial(g\varphi E'_i \rem \Lambda) 
		\leq
		\partial(\psi'_i)
		\leq
		\max\{\partial(g\psi_i),\partial(f_i\varphi)\} - \partial(Y)= B + \partial(\Lambda)
	\end{equation*} 
	which means that $\varphi\in \SE$. Thus, thanks to the hypothesis $\SE=\{0\}$, we get
	$\Lambda|\varphi$ and that  $g\varphi E'_i = \psi'_i = 0 \bmod \Lambda.$
	Thanks to Constraint~\ref{c_3} and the above inequality we can conclude that
	$\partial(\psi'_i)<\partial(\Lambda)$, therefore $\psi'_i=0$ in $\F_q[x]$, \ie 
	\begin{equation}
		\label{eq:samefractions}
		\forall i=1,\ldots,\ell, \ g\psi_i = f_i\varphi.
	\end{equation}
	
	Since $\gcd(f_1,\ldots,f_\ell, g)=1$, Equations~\eqref{eq:samefractions} implies that $g|\varphi$.
	We have already seen that $\Lambda|\varphi$, so $g\Lambda|\varphi$ because $g$ and $\Lambda$ are
	coprime. Plugging $\varphi = a g\Lambda$ for some $a \in \F_q[x]$ into
	Equations~\eqref{eq:samefractions}, we deduce $g\psi_i = f_i\varphi = f_i a g\Lambda$, so
	$\psi_i = a \Lambda f_i$ for all $i$. We have shown $(\varphi,\psi_1,\ldots,\psi_\ell)\in
	(\Lambda g,\Lambda f_1,\ldots,\Lambda f_\ell)\F_q[x]$.
\end{proof}
Thanks to the above lemma we can upper bound the failure probability of
Algorithm~\ref{algoSRF} with $\P_{\fail}\leq \P(\SE\ne\{0\}).$

\begin{remark}\label{rmk:unicity}
	We note that, when the distance parameter $t$ of the decoding algorithm
	satisfies
	$t<\lf\frac{L - d_f - d_g + 1}{2}\rf$ thus (thanks to Lemma~\ref{lm:minDistSRFcode}), it is below half of the minimal distance of the code, we must have that $B + \partial(\Lambda) \leq B + t < 0$ since
	$\partial(\Lambda)\le t$.
	Under such circumstance we therefore
	have $\F_q[x]_{\Lambda, B + \partial(\Lambda)} = \{0\}$, thus 
    $$
    \varphi\in S_{\bE} \Leftrightarrow \forall i = 1,\ldots,\ell, \ \ g\varphi E'_i = 0 \bmod \Lambda.
    $$ We introduce some notation to make the point of this remark: for
     $m\in\F_q[x]$ being any divisor of $\Lambda$, we let
    $(m)\F_q[x]/\Lambda$ denote the set of equivalence classes whose any representative is a multiple of $m$. The above condition it can thus be rephrased as:
    \begin{equation*}
        \varphi \in \left(\frac{\Lambda}{\gcd(gE_1',\ldots, gE'_\ell, \Lambda)}\right)\F_q[x]/\Lambda = \left(\frac{\Lambda}{\gcd\left(g,\Lambda\right)}\right)\F_q[x]/\Lambda,
    \end{equation*}
    where in the last equality we used the hypothesis $\gcd(E'_1,\ldots,E'_\ell, \Lambda) = 1$. Since $\gcd(g,\Lambda) = 1$ we conclude that $S_{\bE} = \{0\}$. Thus the
	failure probability of Algorithm~\ref{algoSRF} is upper bounded by $\P(\SE\ne\{0\}) = 0$, and we get the expected unique decoding result when $t<\lf\frac{L - d_f - d_g + 1}{2}\rf$.
\end{remark}

We study the non-negative random variable $\# \SE$. A standard argument of probability shows that for a discrete non-negative random variable:
$$
\mathbb{E}[\#\SE]=\sum_{m\ge 1}\P(\#\SE\ge m).
$$
Since $0\in \SE$ is always true, we have $\P(\#\SE\ge
1)=1$ and, since $\SE$ is an $\F_q$-vector space, for $2 \le m\le q$, $\P(\#\SE\ge m)=\P(\SE\ne\{0\})$. Thus, we can upper bound 
  $\mathbb{E}[\#\SE] \geq 1+(q-1)\P(\SE\ne\{0\})$.
Therefore, we have $\P(\SE\ne\{0\})\leq (\mathbb{E}[\#\SE]-1) / (q-1)$. Using the expression
$\mathbb{E}[\#\SE]=\sum_{\varphi\in{\F_q[x]}/{\Lambda}}\P\left(\varphi\in \SE\right)$ and $\P(0\in
\SE)=1$, we can write
\begin{align}
	\label{boundInc/ExclSe}
	&\P(\SE\ne\{0\})\leq\frac{1}{q-1}\sum_{\varphi\in ({\mathbb{F}_q[x]}/{\Lambda})\setminus\{0\}}\P\left(\varphi\in \SE\right)\nonumber\\
	&=\frac{1}{q-1}\sum_{\varphi\in({\mathbb{F}_q[x]}/{\Lambda})\setminus\{0\}}\P\left(\forall i, \ g\varphi E_i'\in 
	\F_q[x]_{\Lambda,B + \partial(\Lambda)}\right).
\end{align}

We estimate the terms of the above sum in the following lemma:

\begin{lemma}
	\label{lm:probagpe}
	Assuming Constraint~\ref{c_3},
	if $\varphi\in\F_q[x]$ is such that $\gcd(\varphi,\Lambda)=\eta = \prod_{j\in\xi}(x - \alpha_j)^{\eta_j}$, then
	for the probability distribution of error model $\ERR$, we have that $\P\left(
	   \forall i, \ g \varphi E_i' \in \F_q[x]_{\Lambda,B + \partial(\Lambda)}
	   \right) = 1$ if $\eta = \Lambda$, otherwise 
	\begin{equation*}
        \P\left(
	   \forall i, \ g \varphi E_i' \in \F_q[x]_{\Lambda,B + \partial(\Lambda)}
	   \right)
	   \leq
            \begin{cases}
                \frac{q^{\ell \left(B + 1\right)}
		          }{
		      \prod_{\alpha\in\cZ(\Lambda/\eta)}\left(1 - 1/q^{\ell}\right)            
	           } & \text{if } \ \eta \ne \Lambda,\ \ \partial\left(\Lambda/\eta\right) \geq -B \\
                0 & \text{if } \ \eta \ne \Lambda,\ \ \partial\left(\Lambda/\eta\right) < -B
    \end{cases}	    
	\end{equation*}
\end{lemma}
\begin{proof}
    We start by noticing that if $\eta = \Lambda$ then the probability reduces to $\P\left( \forall i, \ 0 \in \F_q[x]_{\Lambda,B + \partial(\Lambda)}\right) = 1$, thus in what follows we will assume that $\eta \ne \Lambda$.
    
    Since $\gcd(g,M)=1$, the distributions of the vectors $(\varphi E_1',\ldots,\varphi E_\ell')$
	and $(g\varphi E_1',\ldots,g\varphi E_\ell')$ over the sample space $$\Omega_{\Lambda,\ell} \coloneq \{
	(F_i)_{1\le i \le \ell} \in (\F_q[x]/\Lambda)^\ell : \gcd(F_1,\dots,F_\ell,\Lambda) = 1\},$$ are
	identical. Thus, we have 
	$
	\P( \forall i, \ g \varphi E_i' \in \F_q[x]_{\Lambda,B + \partial(\Lambda)} )
	=
	\P( \forall i, \ \varphi E_i' \in \F_q[x]_{\Lambda,B + \partial(\Lambda)} )
	$.
	
	Let us now show that $\varphi E_i' \in \F_q[x]_{\Lambda,B + \partial(\Lambda)} \Leftrightarrow (\varphi/ \eta) E_i' \in
	\F_q[x]_{\Lambda/\eta,B + \partial(\Lambda/ \eta)}$: The first condition can be rephrased as 
	$
	\varphi E_i'=a_i\Lambda + c_i
	$
	with $a_i,c_i\in\F_q[x]$ and $\partial(c_i)\le B + \partial(\Lambda)$. But then we must have that $ \eta|c_i$. Thus, we can divide
	the above by $\eta$ and obtain
	$
	(\varphi/ \eta) E_i'=a_i\Lambda/ \eta + c_i/ \eta
	$
	with $\partial \left(c_i/ \eta\right)\leq B + \partial\left(\Lambda/ \eta\right)$, which is equivalent to $(\varphi/ \eta) E_i' \in \F_q[x]_{\Lambda/
		\eta,B + \partial\left(\Lambda/ \eta\right)}$.
	Since $\gcd\left(\Lambda/\eta,\varphi/\eta\right) = 1$, the distributions of the vectors $(\varphi/ \eta) \left( E_1',\ldots,E_{\ell}'\right)$ and $(E_1',\ldots,E_{\ell}')$ are identical over the sample space $\Omega_{\Lambda/\eta,\ell}$, thus 
    $$
    \P( \forall i, \ \left(\varphi/\eta\right) E_i' \in \F_q[x]_{\Lambda/\eta,B + \partial(\Lambda/\eta)}) = \P( \forall i, \ E_i' \in \F_q[x]_{\Lambda/\eta,B + \partial(\Lambda/\eta)}).
    $$
	When $B + \partial(\Lambda/\eta) < 0$, the previous condition implies that $E_i' = 0 \bmod \Lambda/
	\eta$ for all $i$. Since $ \eta \ne \Lambda$, this is in contradiction with
	$\gcd(E_1',\dots,E_\ell',\Lambda)=1$ for all random matrix $\bE$. Therefore, the associated
	probability $\P( \forall i, \ g \varphi E_i' \in \F_q[x]_{\Lambda,B + \partial(\Lambda)})$ is zero. For the rest of the proof we assume that $B + \partial(\Lambda/\eta) \geq 0$. 
    
	The condition $E_i' \in \F_q[x]_{\Lambda/
		\eta,B + \partial\left(\Lambda/ \eta\right)}$ only depends 
	on the columns $(\be'_j)$ of the reduced random matrix $\bE' = \bE/Y$ for 
	$j \in \xi_{\Lambda/ \eta} \coloneq \{ j : \eta_j < \nu_{\alpha_j}(\Lambda) \}$,
	these columns are uniformly distributed in the sample space 
	$
	\Omega_{\Lambda/\eta,\ell}.
	$
	
	Therefore, letting
	$
	\Upsilon \coloneq \Bigl\{\bE = (\be_j)_{1 \le j \le n} : \forall i, \  E_i' \in \F_q[x]_{\Lambda/
		\eta,B + \partial\left(\Lambda/ \eta\right)} \Bigr\}
	$, we note that $\# \Upsilon = (\# \F_q[x]_{{\Lambda}/{ \eta},{B + \partial\left(\Lambda/\eta\right)}})^\ell  = q^{\ell(B + \partial\left(\Lambda/\eta\right) + 1)}$, where we used the previously stated hypothesis $B + \partial(\Lambda/\eta) \geq 0$.
    We can deduce that our probability equals 
	$$
	\P(\Upsilon) = \frac{ \# (\Omega_{\Lambda/\eta,\ell} \cap \Upsilon) } { \# \Omega_{\Lambda/\eta,\ell} }
	\leq 
	\frac{ \# \Upsilon } { \# \Omega_{\Lambda/\eta,\ell} }.
	$$
	Finally, Lemma~\ref{lm:EulerFormula} tells us that $\# \Omega_{\Lambda/\eta,\ell} = q^{\ell \left(\partial(\Lambda/\eta)\right)}
	\prod_{\alpha\in\cZ(\Lambda/\eta)}\left(1 - 1/q^{\ell}\right)  $.
\end{proof}
Before proving our results we still need the following technical lemma.

\begin{lemma}
	\label{lm:sumOverDivisorsPOLY}
	Fix $\Lambda \in\F_q[x]$ having all its roots in $\F_q$ and $f(x)$ an arbitrary real-valued function such that $f(0) = 1$. Then 
	\begin{equation*}
		\sum_{\eta|\Lambda}
		\prod_{\alpha\in\cZ(\eta)}
		f(\nu_{\alpha}(\eta))
		=
        \prod_{\alpha\in\cZ(\Lambda)}\left[1 + \sum_{k= 1}^{\nu_{\alpha}(\Lambda)}f(k)\right]
		\end{equation*}
\end{lemma}
\begin{proof}
    Since $f(0) = 1$, the product on the right hand side ca be written as $\prod_{\alpha\in\cZ(\Lambda)}
		\sum_{k= 0}^{\nu_{\alpha}(\Lambda)}f(k)$ and then expanded as
	\begin{equation*}
		\prod_{\alpha\in\cZ(\Lambda)}
		\sum_{k= 0}^{\nu_{\alpha}(\Lambda)}f(k)
		=
		\sum_{\substack{\left(\eta_{\alpha}\right)_{\alpha\in \cZ(\Lambda)}\\ 0\le\eta_{\alpha}\le\nu_{\alpha}(\Lambda)}}
		\prod_{\alpha\in \cZ(\Lambda)}
		f(\eta_\alpha).
	\end{equation*}
    We recognize in the right hand side sum a sum over the divisors $\eta$ of $\Lambda$ with $\nu_{\alpha}(\eta) = \eta_{\alpha}$, thus
    \begin{equation*}
        \sum_{\substack{\left(\eta_{\alpha}\right)_{\alpha\in \cZ(\Lambda)}\\ 0\le\eta_{\alpha}\le\nu_{\alpha}(\Lambda)}}
		\prod_{\alpha\in \cZ(\Lambda)}
		f(\eta_\alpha)
        =
        \sum_{\eta|\Lambda}
		\prod_{\alpha\in\cZ(\Lambda)}
		f(\nu_{\alpha}(\eta))
        =
        \sum_{\eta|\Lambda}
		\prod_{\alpha\in\cZ(\eta)}
		f(\nu_{\alpha}(\eta))
    \end{equation*}
    where in the last equality we used the hypothesis that $f(0) = 1$.
\end{proof}

\begin{lemma}
	\label{lm:boundfailureprobability}
	Given a random vector $(E_1',\ldots,E_\ell')$ uniformly distributed in $\Omega_{\Lambda,\ell}$, we have that
	
	\begin{equation*}
		\P
		\left(
			\SE \ne\{0\}
		\right) 
		\leq
		\frac{q^{\ell\left(B + 1\right)}}{q - 1}
		 q^{\partial(\Lambda)}\prod_{\alpha\in\cZ(\Lambda)}\left(\frac{1 - 1/q^{\ell + \nu_{\alpha}(\Lambda)}}{1 -1 / q^{\ell}}\right).
	\end{equation*}
\end{lemma}
\begin{proof}
    Thanks Lemma~\ref{lm:probagpe}, the sum in~\eqref{boundInc/ExclSe} can be temporarily reduced to nonzero elements $\varphi\in \F_q[x]/\Lambda$ such that $\partial\left(\gcd\left(\varphi,\Lambda\right)\right)\leq B + \partial(\Lambda)$, giving:
\begin{align}
    \label{sumWithCondition}
    \P(\SE\ne\{0\})
    &\leq\frac{1}{q-1}\sum_{\substack{\varphi\in ({\mathbb{F}_q[x]}/{\Lambda})\setminus\{0\}\\ \partial\left(\gcd\left(\varphi,\Lambda\right)\right)\leq B + \partial(\Lambda)}}
    \frac{q^{\ell \left(B + 1\right)}
		          }{
		      \prod_{\alpha\in\cZ(\Lambda/\eta)}\left(1 - 1/q^{\ell}\right)            
	           }\nonumber\\
    &\leq \frac{1}{q-1}\sum_{\varphi\in ({\mathbb{F}_q[x]}/{\Lambda})\setminus\{0\}}
    \frac{q^{\ell \left(B + 1\right)}
		          }{
		      \prod_{\alpha\in\cZ(\Lambda/\eta)}\left(1 - 1/q^{\ell}\right)            
	           }          
\end{align}
	Since the terms in the sum depend only on $\eta$, we regroup the $\varphi$ in the sum by their gcd with $\Lambda$.  
	Note that, thanks to Lemma~\ref{lm:EulerFormula}, the number of elements $\varphi\in \F_q[x]/\Lambda$ such that $\gcd(\varphi,\Lambda)=\eta$, is given by $\#\Omega_{\Lambda/\eta,1} = q^{ \partial(\Lambda/\eta)}
	\prod_{\alpha\in\cZ(\Lambda/\eta)}\left(1 - 1/q\right)$.  
	Therefore, extending the sum over all the divisors of $\Lambda$
	\begin{equation*}
        \P(\SE\ne\{0\})\leq
		\frac{q^{\ell(B+1)}}{q - 1}
		\sum_{\eta|\Lambda}
		q^{ \partial(\Lambda/\eta)}
		\prod_{\alpha\in\cZ(\Lambda/\eta)}
		\frac{\left(1 - 1/q\right)}{\left(1 - 1/q^\ell\right)}
	\end{equation*}	
	we can upper bound the quotient $\P(\SE\ne\{0\})/\frac{q^{\ell(B+1)}}{q - 1}$ with
	\begin{align*}
		\sum_{\eta|\Lambda}
		\prod_{\alpha\in\cZ(\eta)}
		\frac{1 - \frac{1}{q}}{1 - \frac{1}{q^{\ell}}}
		q^{\nu_{\alpha}\left(\eta\right)}
		=
		\prod\limits_{\alpha \in\cZ(\Lambda)}
		\left(
		1+
		\frac{1 - \frac{1}{q}}{1 - \frac{1}{q^{\ell}}}
		\sum_{k = 1}^{\nu_{\alpha}(\Lambda)}
		q^{k}
		\right)
	\end{align*}
	where in the last equality we used Lemma~\ref{lm:sumOverDivisorsPOLY} with 
	\begin{equation*}
		f(x)=
        \begin{cases}
		  \frac{1 - \frac{1}{q}}{1 - \frac{1}{q^\ell}}q^{x} & \text{if }\ x > 0\\
            0 & \text{if }\ x = 0
		\end{cases}
	\end{equation*}
	To conclude we notice that
	\begin{equation*}
		\prod\limits_{\alpha \in\cZ(\Lambda)}
		\left(
		1+
		\frac{1 - \frac{1}{q}}{1 - \frac{1}{q^{\ell}}}
		\sum_{k = 1}^{\nu_{\alpha}(\Lambda)}
		q^{k}
		\right)
		=
		\prod_{\alpha\in\cZ(\Lambda)}
		\frac{q^{\nu_{\alpha}(\Lambda)} - 1 / q^{\ell}}{1 - 1 / q^\ell}
		=
		q^{\partial(\Lambda)}
		\prod_{\alpha\in\cZ(\Lambda)}
		\frac{1 - 1 / q^{\ell + \nu_{\alpha}(\Lambda)}}{1 - 1 / q^\ell}. \qedhere
	\end{equation*}
\end{proof}		

\begin{proof}[Proof of Theorem~\ref{thm:main1}]
	We start by noticing that for every $\ell>0$, any choice of the input parameter $t\leq \tm$ satisfies  Constraint~\ref{c_3}, thus we can apply all the previous lemmas and upper bound the failure probability of Algorithm~\ref{algoSRF} with the quantity given by Lemma~\ref{lm:boundfailureprobability}.
	
	Thanks to the hypothesis of Theorem~\ref{thm:main1} we know that $\partial(\Lambda) \leq t$, 
	and using 
	$q^{\ell\left(B + 1\right)} q^t  = q^{-(\ell + 1)(\tm - t)}$, 
	we have proved Theorem~\ref{thm:main1}.
\end{proof}

\subsubsection{Decoding failure probability with respect to the second error model}
\label{RN_ERR2}

In the second error model, we need to make a distinction between the maximal error locator $\Lambda_m$ (over which there are uniform random errors) and the actual error locator $\Lambda = \Lambda_{\bE}$ which can be a proper divisor of $\Lambda_m$. We will denote $\P_{\ERRd}$ (resp. $\P_{\ERR}$) the probability function under the error model 2 (resp. the error model 1). 
Let $\cF$ be the event of decoding failure with algorithm parameter $t\geq\partial(\Lambda_m)$ \ie\  the set of random matrices $\bE$ such that Algorithm~\ref{algoSRF} returns "decoding failure".
Using the law of total probability, we have
\begin{equation}
	\label{LawTotalProba}
	\P_{\ERRd}( \cF )
	=
	\sum_{\Lambda|\Lambda_m}
	\P_{\ERRd}(\cF \ |\ \Lambda_{\bE} = \Lambda)
	\ 
	\P_{\ERRd}( \Lambda_{\bE} = \Lambda).
\end{equation}
The conditional probabilities 
$\P_{\ERRd}(\cF \ |\ \Lambda_{\bE} = \Lambda)$
in the sum  are equal to 
$\P_{\ERR}(\cF)$, which are upper bounded within the proof of Lemma~\ref{lm:boundfailureprobability} by
\begin{equation}
	\label{IntermediateBuondERR1}
	\P_{\ERR}(\cF)
	\leq
	\frac{q^{\ell\left(B + 1\right)}}{q - 1}
	q^{\partial(\Lambda)}\prod_{\alpha\in\cZ(\Lambda)}\left(\frac{1 - 1/q^{\ell + \nu_{\alpha}(\Lambda)}}{1 -1 / q^{\ell}}\right).
\end{equation}
Moreover, using again Lemma~\ref{lm:EulerFormula}, we have 
\begin{equation}
	\label{ConditionProba}
	\P_{\ERRd}( \Lambda_{\bE} = \Lambda)
	=
	\frac{\#\Omega_\Lambda^\ell}{q^{\ell\partial(\Lm)}}
	=
	\frac{q^{\ell\partial(\Lambda)}}{q^{\ell\partial(\Lm)}}
	\prod_{\alpha\in\cZ(\Lambda)}\left(1 - \frac{1}{q^{\ell}}\right).
\end{equation}
Using these facts we can prove Theorem~\ref{thm:main2}.

\begin{proof}[Proof of Theorem~\ref{thm:main2}]
	Plug Equations~\eqref{IntermediateBuondERR1} and~\eqref{ConditionProba} in Equation~\eqref{LawTotalProba} to obtain that 
	$\P_{\ERRd}( \cF ) / \frac{q^{\ell\left(B + 1\right)}}{(q - 1)q^{\ell\partial(\Lm)}}$ is less than or equal to
	\begin{align*}
		\sum_{\Lambda|\Lambda_m}
		q^{(\ell + 1)\partial(\Lambda)}
		\prod_{\alpha \in\cZ(\Lambda)}\left(1 -\frac{1}{q^{\ell + \nu_{\alpha}(\Lambda)}}\right) 
		&=
		\sum_{\Lambda|\Lambda_m}
		\prod_{\alpha\in\cZ(\Lambda)}q^{\nu_{\alpha}(\Lambda)(\ell + 1)}\left(1 -\frac{1}{q^{\ell + \nu_{\alpha}(\Lambda)}}\right)\\
		&=
		\prod_{\alpha\in\cZ(\Lambda_m)}\left[1 + \sum_{k = 1}^{\nu_{\alpha}(\Lambda_m)}q^{k(\ell + 1)}\left(1 -\frac{1}{q^{\ell + k}}\right)\right]\\
		&\le
		\prod_{\alpha\in\cZ(\Lambda_m)}\left[1 + \left(1 -\frac{1}{q^{\ell + \nu_{\alpha}(\Lambda_m)}}\right)\sum_{k = 1}^{\nu_{\alpha}(\Lambda_m)}q^{k(\ell + 1)}\right],
	\end{align*}
	where we used again Lemma~\ref{lm:sumOverDivisorsPOLY} with 
    \begin{equation*}
        f(x)
        =
        \begin{cases}
            q^{x(\ell + 1)}\left(1 -\frac{1}{q^{\ell + x}}\right) &\text{ if }\ x > 0\\
            0    &\text{ if }\ x = 0
        \end{cases}
    \end{equation*}
	and in the last inequality we used that $1 -1/q^{\ell + k}\leq 1 -1/q^{\ell + \nu_{\alpha}(\Lambda_m)}$ for every $k=1,\ldots,\nu_{\alpha}(\Lambda_m)$. 
	By computing the geometric sum inside the last product, the above is equal to
	\begin{align*}
		&\prod_{\alpha\in\cZ(\Lambda_m)}
		\left[
		1 + \left(1 - \frac{1}{q^{\ell + \nu_{\alpha}(\Lambda_m)}}\right)
		\left(\frac{q^{(\ell + 1)(\nu_{\alpha}(\Lambda_m) + 1)} - q^{\ell + 1}}{q^{\ell + 1} - 1} \right)
		\right]\\
        &=
		\prod_{\alpha\in\cZ(\Lambda_m)}
		\left[
		1 +
		\frac{1 - 1 / q^{\ell + \nu_{\alpha}(\Lambda_m)}}{1 - 1 /q^{\ell + 1}}
		\left(q^{\nu_{\alpha}(\Lambda_m)(\ell + 1)} - 1\right)
		\right].
	\end{align*}
	Since $\nu_{\alpha}(\Lambda_m) \geq 1$ we can upper bound the first $1$ in each term as  
	$$
	1 \leq \frac{1 - 1 / q^{\ell + \nu_{\alpha}(\Lambda_m)}}{1 - 1 /q^{\ell + 1}}
	$$
	and the above product is upper bounded as:
	\begin{equation*}
		\prod_{\alpha\in\cZ(\Lambda_m)}
		\left[
		1 +
		\frac{1 - 1/q^{\ell + \nu_{\alpha}(\Lambda_m)}}{1 - 1/q^{\ell + 1}}
		\left(q^{\nu_{\alpha}(\Lambda_m)(\ell + 1)} - 1\right)
		\right]
		\le
		q^{\partial(\Lm)(\ell + 1)}
		\prod_{\alpha\in\cZ(\Lambda_m)}
		\frac{1 - 1/q^{\ell + \nu_{\alpha}(\Lambda_m)}}{1 - 1/q^{\ell + 1}}
	\end{equation*}

	Now, thanks to the hypothesis of the theorem 
	we know that $\partial(\Lm)\leq t$, thus we can write 
	\begin{align*}
		\P_{\ERRd}( \cF ) &\leq \frac{q^{\ell\left(B + 1\right)}}{q - 1}	q^{\partial(\Lm)}
		\prod_{\alpha\in\cZ(\Lambda_m)}
		\frac{1 - 1/q^{\ell + \nu_{\alpha}(\Lambda_m)}}{1 - 1/q^{\ell + 1}}\\
		&\leq \frac{q^{\ell\left(B + 1\right)}}{q - 1}	q^{t}
		\prod_{\alpha\in\cZ(\Lambda_m)}
		\frac{1 - 1/q^{\ell + \nu_{\alpha}(\Lambda_m)}}{1 - 1/q^{\ell + 1}}.
	\end{align*}
	Using $q^{-(\ell+1)(\tm - t)}=q^{\ell\left(B + 1\right)} q ^t$,
	we have proved Theorem~\ref{thm:main2}.
\end{proof}

\section{Analysis of the decoder for a hybrid error model}
\label{sec:Hybrid Decoding}

In this section we consider a hybrid approach to the failure probability analysis for the multiplicity rational function codes studied above.
The approach is hybrid in the sense that it lies in between unique decoding and interleaving.

As in the previous section, the algorithm parameter $t$ is chosen and the probability of failure is studied under a random distribution of received words at distance at most $t$ from a given codeword associated to a reduced vector of fractions $\bf/g$.
With the hybrid approach, we analyze the failure probability for a modified random distribution, still determined by the algorithm parameter, but in which $t$ splits into two components: $t_{i}$ and $\htu$. While $\htu$ captures constant errors of the distribution and is bounded to fit the unique decoding, $t_{i}$ can be larger as it incorporates random errors and reflects the interleaving decoding technique having its bound $\tmi$ (Equation~\eqref{def:DmaxHyb}) increasing with $\ell$ beyond unique capability and asymptotically up to $L - d_f - d_g + 1 - 2\htu$.
Notably, if $t_{i}=0$, the algorithm never fails as $t = \htu \leq \lf\frac{L - d_f - d_g +1}{2} \rf$. Therefore, the probability of decoding failure is strictly related to $t_{i}$ and is analyzed under probabilistic assumptions, considering a random distribution of errors. 

The motivation for splitting $t$ is that not all errors can be assumed to be purely random. 
For instance, in the context of distributed computation, some errors might be introduced by malicious entities that deliberately choose specific error patterns to force the algorithm to fail. In such cases, the errors captured by $\htu$ remain independent of the error distribution and can still be corrected.

Since we are above the unique decoding radius, not all errors are decodable. Interleaving techniques can provide positive decoding results by considering error sets where most errors are decodable using probabilistic arguments. These techniques focus on fixed error positions and consider all possible errors at each position.
In contrast, in a hybrid setting, one can handle more general sets of errors by analyzing the set of all possible errors across certain subsets of the error positions.
This approach may be of broader interest in coding theory.

A first particular instance of this hybrid model was introduced in~\cite{guerrini2023simultaneous} in the same context studied in this paper but with a different analysis, then it has been jointly generalized in~\cite{abbondati2025simultaneous,brakensiek2025unique}.  
We remark that, as in~\cite{guerrini2023simultaneous}, in the forthcoming case of codes allowing poles (see Section~\ref{Sec:poles}), we are only able to perform interleaving on a subset of all errors (namely evaluation errors)
This suggests that there could be a deeper obstacle preventing the interleaving of the other type of errors (namely valuation errors).

On a technical level this hybrid analysis consists in studying the failure probability with respect to a specific portion of the error's distribution; allowing the errors to vary only over a subset $\xii\subseteq\xi$ of the  error support, while the errors in the complementary set $\xiu\coloneq\xi\setminus\xii$ are held fixed. 
Note that, in this section the above partition might seem arbitrary but, as we will see in the next Section~\ref{Sec:poles} on poles, it is clearly described by some property of the error itself (see Definition~\ref{def:errorsupportbad primes}).
Here we generalize the analysis of the previous section relative to the decoding of $\SRF$ codes (Definition~\ref{def:SRF}) by means of Algorithm~\ref{algoSRF}.
In this setting we decompose the distance parameter $t$ of the algorithm as
\begin{equation}
	\label{distParamAlgoHYB}
	t = \hti + \htu,
\end{equation}
for some  $\hti,\htu \geq 0$ bounds on the sizes of random and fixed errors respectively.

\paragraph*{Error models}
With the given distance parameter $t$ as in Equation~\eqref{distParamAlgoHYB}, 
we perform the hybrid analysis with respect to a distribution
specified by a  factorization of $\Lambda = \Lu\Li$ with $\gcd(\Lu,\Li) = 1$, where $\Lambda$ divides $M$. To specify the error model, we fix a sequence of nonzero error vectors 
$\beps_j\in\left(\F_q[x]/(x - \alpha_j)^{\lambda_j}\right)^{\ell}$ for every $j$ 
such that $\alpha_j\in\cZ(\Lu)$, with $\nu_{\alpha_j}(\beps_j) = \lambda_j - \nu_{\alpha_j}(\Lambda)$. 
Then the random distribution for the hybrid error model is determined by the set of error matrices 
$\bE\in\prod_{j=1}^{n}\left(\F_q[x]/(x - \alpha_j)^{\lambda_j}\right)^{\ell}$ such that the columns $\be_j$ of $\bE$ satisfy
\begin{enumerate}
	\item $\be_j = \boldsymbol{0} \text{ for all } j \text{ such that } \alpha_j\not\in\cZ(\Lambda)$,
	\item $\be_j = \beps_j \ \text{ for all } j\text{ such that } \alpha_j\in\cZ(\Lu)$,
	\item $\nu_{\alpha_j}(\be_j) = \lambda_j - \nu_{\alpha_j}(\Lambda) \text{ for all } j \text{ such that } \alpha_j\in\cZ(\Li)$.
\end{enumerate}
We let $\HYB$ be the set of error matrices specified as above.

\begin{lemma}
	\label{lm:randomuniform}
	If $\bE$ is uniformly distributed in $\HYB$, then 
	the random vector $(E_1' \bmod \Li,\ldots,E_\ell' \bmod \Li)$ is uniformly
	distributed in the sample space $\Omega_{\Li,\ell}$.
\end{lemma}
\begin{proof}
	For the duration of this proof, we will only consider indices $j$ such that
	$\alpha_j\in\cZ(\Li)$, so that $\nu_{\alpha_j}(\Lambda) = \nu_{\alpha_j}(\Li)$. Recall that  $\be_{j}$ is a random vector of 
	$\left(\F_q[x]/(x - \alpha_j)^{\lambda_j}\right)^{\ell}$ of valuation $\lambda_j - \nu_{\alpha_j}(\Lambda)$
	for all those particular $j$. Since $\nu_{\alpha_j}(Y) = \nu_{\alpha_j}(\be_{j}) = \lambda_j - \nu_{\alpha_j}(\Lambda)$,
	we get that the vector $\be_{j}/Y \in
	(\F_q[x]/(x - \alpha_j)^{\nu_{\alpha_j}(\Lambda)})^\ell$ is random of valuation $0$. As a
	consequence, by the definition of $\HYB$, we obtain that $(E_1' \bmod \Li,\ldots,E_\ell' \bmod \Li)$ is random
	among the vectors of $(\F_q[x]/\Li)^\ell$ such that
	$\gcd(E_1',\ldots,E_{\ell}',\Li) = 1$.
\end{proof}

\medskip

As for the hybrid version of the error model $\ERRd$, we fix a maximal error locator $\Lambda_m$ factorized as $\Lambda_m = \Lmi\Lu$ with $\gcd\left(\Lmi,\Lu\right) = 1$, where $\Lm$ divides $M$. We fix a sequence of nonzero error vectors $\beps_j\in\left(\F_q[x]/(x - \alpha_j)^{\lambda_j}\right)^{\ell}$ for every $j$ such that $\alpha_j\in\cZ(\Lu)$, with $\nu_{\alpha_j}(\beps_j) = \lambda_j - \nu_{\alpha_j}(\Lambda_m)$. Then we consider the set of error matrices $\bE\in\prod_{j=1}^{n}\left(\F_q[x]/(x - \alpha_j)^{\lambda_j}\right)^{\ell}$ such that
\begin{enumerate}
	\item $\be_j = \boldsymbol{0} \text{ for all } j \text{ such that } \alpha_j\not\in\cZ(\Lambda_m)$,
	\item $\be_j = \beps_j \ \text{ for all } j\text{ such that } \alpha_j\in\cZ(\Lu)$,
	\item $\nu_{\alpha_j}(\be_j) \ge \lambda_j - \nu_{\alpha_j}(\Lambda_m) \text{ for all } j \text{ such that } \alpha_j\in\cZ(\Lmi)$.
\end{enumerate}
We let $\HYBd$ be the set of error matrices specified as above.

We notice that for a given error matrix $\bE$ in the distribution $\HYBd$ the associated error locator has the form $\Lambda_{\bE} = \Li\Lu$ for some divisor $\Li|\Lmi$.

\paragraph*{Our results} We can now state our results concerning the analysis of the correctness of the decoder \wrt to a hybrid error model. Define
\begin{equation}
	\label{def:DmaxHyb}
	\tmi \coloneq 
	\frac{\ell}{\ell+1} \left[L - d_f - d_g + 1 - 2 \htu\right].
\end{equation}

Note that we must have $2 \htu \le L - d_f - d_g + 1$ in order to ensure $\tmi \ge 0$.

\begin{theorem}
	\label{thm:main1hyb}
	Decoding Algorithm~\ref{algoSRF} on input
	\begin{enumerate}
		\item distance parameter $t = \htu + \hti$ for $\htu\leq \lf\frac{L - d_f - d_g + 1}{2}\rf$ and $\hti\leq \tmi$,
		\item a random received word $\bR$ uniformly distributed in $\bC + \HYB$
		for some code word $\bC  \in \SRF_\ell(M;d_f,d_g)$ and error locator $\Lambda = \Li\Lu$ such that $\partial (\Lu) \leq \htu$ and $\partial(\Li)\leq \hti$,
	\end{enumerate}
	outputs the center  code word $\bC$ of the distribution
	with a probability of failure 
	\begin{equation*}
		\P_{\fail}
		\leq 
		\frac{q^{-(\ell +1)(\tmi-\hti)}}{q-1}
		\prod_{\alpha\in\cZ(\Li)}\left(\frac{1 - 1/q^{\ell + \nu_\alpha(\Li)}}{1 -1/q^{\ell}}\right).
	\end{equation*} 
\end{theorem}
\begin{theorem}
	\label{thm:main2hyb}
	Decoding Algorithm~\ref{algoSRF} on input
	\begin{enumerate}
		\item distance parameter $t = \htu + \hti$ for $\htu\leq \lf\frac{L - d_f - d_g + 1}{2}\rf$ and $\hti\leq \tmi$,
		\item a random received word $\bR$ uniformly distributed in  $\bC  + \HYBd$
		for some code word $\bC \in \SRN_\ell(M;d_f,d_g)$ and error locator $\Lambda_m = \Lmi\Lu$ such that $\partial (\Lu) \leq \htu$ and $\partial(\Lmi)\leq \hti$,
	\end{enumerate}
	outputs the center  code word $\bC$ of the distribution
	with a probability of failure 
	\begin{equation*}
		\P_{\fail}
		\leq 
		\frac{q^{-(\ell +1)(\tmi-\hti)}}{q-1}
		\prod_{\alpha\in\cZ(\Lmi)}
		\left(\frac{1 - 1/q^{\ell + \nu_{\alpha}(\Lmi)}}{1 - 1/q^{\ell + 1}}\right).
	\end{equation*} 
\end{theorem}

\begin{exam}\label{example}
	Let's give a scenario that would highlight how Theorem~\ref{thm:main2hyb}
	can be used in practice. Assume that a code is fixed such that 
	$ L - d_f - d_g + 1 = 200$, so that $\tm = 160$ when one interleaves for $\ell=4$.
	Assume one wanted to make sure that the failure probability is less than a target probability of $q^{-31}$, and also that $50$ weighted errors can always be corrected ($\htu = 50$), for instance for protecting against a malicious entity.
	Then $\tmi = 80$ and one would have to choose the parameter $t = 134$ (thus $\hti = 74$) for the decoder (where we approximate the failure probability by $q^{-(\ell+1)(\tmi - \hti)}/q$).
	Then Theorem~\ref{thm:main2hyb} would ensure that for any error with locator 
	$\Lu$ such that $\partial(\Lu) \le 50$ and for any random error distributed uniformly
	on an error locator $\Lmi$ such that $\partial(\Lmi) \le 74$ (with $\Lmi$ and $\Lu$ coprime), the failure probability is less than $q^{-31}$.
\end{exam}

We introduce a modified version of the set $\SE$ defined as
\begin{equation*}
	\SEh\coloneq
	\left\{
	\varphi\in\F_q[x]/\Li : \forall i, \ g\varphi E'_i \in \F_q[x]_{\Li,B + \partial(\Lambda)}
	\right\}
\end{equation*}
with $B\coloneq d_f + d_g + t - L - 2 = d_f + d_g + \hti +\htu - L - 2 $. 
The hybrid versions of Constraint~\ref{c_3} and Lemma~\ref{BaseLemmaRN} are as follows:

\begin{constraint}
	\label{cst: ParamHybridVers}
	The parameters of Algorithm~\ref{algoSRF} satisfy $B + \htu < 0$.
\end{constraint}
\begin{lemma}
	\label{BaseLemmaRNHybrid}
	If Constraint~\ref{cst: ParamHybridVers} is satisfied then $\SEh = \{0\}
	\Rightarrow S_{\bR}\subseteq v_{\bC}\F_q[x]$.
\end{lemma}
\begin{proof}
	Let $(\varphi,\psi_1,\ldots,\psi_\ell)\in S_{\bR}$. 
	The proof of Lemma~\ref{BaseLemmaRN} shows that $g\varphi E'_i$ is equal to 
	$\psi'_i \coloneq \frac{g\psi_i-f_i\varphi}{Y}$ modulo $\Lambda$, hence also modulo $\Li$.
	The same proof gives $\partial(\psi'_i)  \leq B + \partial(\Lambda)$.
	This means that $\varphi\in \SEh$, thus thanks to the hypothesis $\SEh=\{0\}$, we get
	$\Lambda_i|\varphi$, thus $g\varphi E'_i = \psi'_i = 0 \bmod \Li.$
	Thanks to Constraint~\ref{cst: ParamHybridVers} we have that $\partial(\psi'_i)  \leq B + \partial(\Lambda) < \partial(\Li)$, therefore $\psi'_i=0$ in $\F_q[x]$.
	The end of the proof is identical to the one of Lemma~\ref{BaseLemmaRN}.
\end{proof}

As in Equation~\eqref{boundInc/ExclSe}, we have 
$$
\P(\SEh\ne\{0\})
\leq
\frac{1}{q-1}\sum_{\varphi\in({\mathbb{F}_q[x]}/{\Li})\setminus\{0\}}\P\left(\forall i, \ g\varphi E_i'\in 
\F_q[x]_{\Li,B + \partial(\Lambda)}\right), 
$$
whose sum we now bound.

\begin{lemma}
	\label{lm:boundfailureprobabilityhybrid}
	Given a random vector $(E_1',\ldots,E_\ell')$ uniformly distributed in $\Omega_{\Li,\ell}$, we have that
	
	\begin{align*}
	\sum_{\varphi\in({\mathbb{F}_q[x]}/{\Li})\setminus\{0\}}\P\left(\forall i, \ g\varphi E_i'\in 
		\F_q[x]_{\Li,B + \partial(\Lambda)}\right)
		\leq
		q^{\ell(B +1 +\htu )}
		 q^{\partial(\Li)}
		\prod_{\alpha\in\cZ(\Li)}\left(\frac{1 - 1 / q^{\ell + \nu_{\alpha}(\Li)}}{1 -1 / q^{\ell}}
		\right).
	\end{align*}
\end{lemma}
\begin{proof}
	As in the proof of Lemma~\ref{lm:probagpe}, assuming Constraint~\ref{cst: ParamHybridVers} we can upper bound the generic term of the above sum over $\varphi\in ({\mathbb{F}_q[x]}/{\Li})\setminus\{0\}$ in terms of $\eta\coloneq\gcd(\Li,\varphi)$ as
	$$
		\P\left(\forall i, \ g\varphi E_i'\in 
		\F_q[x]_{\Li,B + \partial(\Lambda)}\right)
		\leq
		\frac{q^{\ell(B+1+\partial(\Li/\eta)+\partial(\Lu))}}{\# \Omega_{\Li,\eta}}
		=
		\frac{q^{\ell(B+1+\partial(\Lu))}}{\prod_{\alpha\in\cZ(\Li/\eta)}\left(1 - 1/q^\ell\right)}.	
	$$
	Thus, by considering that for each gcd $\eta$ the number of terms in the sum is  $\#\Omega_{\Li/\eta,1} = q^{\partial(\Li/\eta)}\prod_{\alpha\in\cZ(\Li/\eta)}\left(1 - 1/q\right)$ we can upper bound the sum as
	\begin{align*}
		q^{\ell(B+1 + \partial(\Lu))}
		\sum_{\eta|\Li}
		q^{ \partial(\Li/\eta)}
		\prod_{\alpha\in\cZ(\Li/\eta)}
		\frac{\left(1 - 1/q\right)}{\left(1 - 1/q^\ell\right)}
		& \leq 
	q^{\ell(B+1 + \partial(\Lu))}
	q^{\partial(\Li)}\prod_{\alpha\in\cZ(\Li)}\left(\frac{1 - 1/q^{\ell + \nu_{\alpha}(\Li)}}{1 -1 / q^{\ell}}\right).
	\end{align*}
    where the last inequality is obtained, as in the proof of Lemma~\ref{lm:boundfailureprobability}, by using Lemma~\ref{lm:sumOverDivisorsPOLY} with 
	\begin{equation*}
		f(x)=
        \begin{cases}
		  \frac{1 - \frac{1}{q}}{1 - \frac{1}{q^\ell}}q^{x} & \text{if }\ x > 0\\
            0 & \text{if }\ x = 0
		\end{cases}
	\end{equation*}
    and computing the resulting geometric sum.
	Using that $\partial(\Lu)\leq \htu$ we obtain our statement.
\end{proof}

\begin{proof}[Proof of Theorem~\ref{thm:main1hyb}]
	As in the proof of Theorem~\ref{thm:main1}, we start by noticing that our choice of parameters satisfy Constraint~\ref{cst: ParamHybridVers}. We first notice that $\hti\leq \tmi = \ell/(\ell+1) \left[L - d_f - d_g + 1 - 2 \htu\right] < L - d_f - d_g + 1 - 2 \htu$, thus Constraint~\ref{cst: ParamHybridVers} is satisfied.
	Thanks to Lemma~\ref{BaseLemmaRNHybrid} and Lemma~\ref{lm:boundfailureprobabilityhybrid}, we can upper bound the failure probability by
	\begin{equation*}
		\P_{fail}
		\leq
		\P(\SEh\ne\{0\})
		\leq
        \frac{q^{\ell(B +1 +\htu )}}{q - 1}
		q^{\partial(\Li)}
		\prod_{\alpha\in\cZ(\Li)}\left(\frac{1 - 1 / q^{\ell + \nu_{\alpha}(\Li)}}{1 -1 / q^{\ell}}
		\right).
	\end{equation*}
	Since $\partial(\Li)\leq \hti$, we have
	$
	q^{\ell(B +1 +\htu )}
	q^{\partial(\Li)}
	\le 
	q^{\ell(B +1 +\htu )}
	q^{\hti}
	=
	q^{-(\ell+1)(\tmi - \hti)}
	$.
\end{proof}

\begin{proof}[Proof of Theorem~\ref{thm:main2hyb}]
	Let $\cF$ be the event of decoding failure, \ie{}the set of random matrices $\bE$ such that Algorithm~\ref{algoSRF} returns "decoding failure" with input parameter $t = \hti + \htu$ as in the statement of Theorem~\ref{thm:main2hyb}.
	We will denote $\P_{\HYBd}$ (resp. $\P_{\HYB}$) the probability function under the hybrid error model 2 (resp. model 1) specified by a given factorization of the error locator, and by a sequence of fixed error vectors $\beps_j$ for every $j$ such that $\alpha_j\in\cZ\left(\Lu\right)$.

	Using the law of total probability, we have that $\P_{\HYBd}( \cF )$ can be decomposed as the sum
	
	$$
	\P_{\HYBd}(\cF)
	=
	\sum_{\Li|\Lmi}
	\P_{\HYBd}(\cF \ |\ \Lambda_{\bE} = \Li\Lu)
	\ 
	\P_{\HYBd}( \Lambda_{\bE} = \Li\Lu),
	$$
	where 
	$
	\P_{\HYBd}(\cF \ |\ \Lambda_{\bE} = \Li\Lu)
	=
	\P_{\HYB}(\cF)
	$,
	whereas
	$$
	\P_{\HYBd}(\Lambda_{\bE} = \Li\Lu)
	= 
	\frac{q^{\ell\partial(\Li)}}{q^{\ell\partial(\Lmi)}}
	\prod_{\alpha\in\cZ(\Li)}\left(1 - \frac{1}{q^{\ell}}\right)
	$$
	as in Equation~\eqref{ConditionProba}.
	
	Plugging the above two expressions in the decomposition from the law of total probability, similarly as done in the proof of Theorem~\ref{thm:main2}, we can upper bound 
	$
	\P_{\HYBd}(\cF)
    /\frac{q^{\ell\left(B + 1 + \htu\right)}}{(q - 1)q^{\ell\partial(\Lmi)}}
	$ 
	by
	\begin{equation*}
		\sum_{\Li|\Lmi}
		q^{(\ell + 1)\partial(\Li)}
		\prod_{\alpha \in\cZ(\Li)}\left(1 -\frac{1}{q^{\ell + \nu_{\alpha}(\Li)}}\right)
		\leq
		q^{(\ell + 1)\partial(\Lmi)}
		\prod_{\alpha\in\cZ(\Lmi)}
		\frac{1 - 1/q^{\ell + \nu_{\alpha}(\Lmi)}}{1 - 1/q^{\ell + 1}}.
	\end{equation*}
	Thus,
	\begin{align*}
		\P_{\HYBd}( \cF )
		&\leq
		\frac{q^{\ell(B +1 +\htu)}}{q - 1}  q^{\partial(\Lmi)}
		\prod_{\alpha\in\cZ(\Lmi)}
		\frac{1 - 1/q^{\ell + \nu_{\alpha}(\Lmi)}}{1 - 1/q^{\ell + 1}}\\
		&\leq 
		\frac{q^{\ell(B +1 +\htu)}}{q - 1}  q^{\hti}
		\prod_{\alpha\in\cZ(\Lmi)}
		\frac{1 - 1/q^{\ell + \nu_{\alpha}(\Lmi)}}{1 - 1/q^{\ell + 1}}
	\end{align*}
	and we conclude by using that 
	$
	q^{\ell(B +1 +\htu)}  q^{\hti}
	=
	q^{-(\ell+1)(\tmi - \hti)}
	$.
\end{proof}

\section{The case of poles}
\label{Sec:poles}
In this section we use the hybrid analysis technique presented above to extend our study of the decoding failure in a context where the hypothesis $\gcd(g,M) = 1$ of Definition~\ref{def:SRF} does not hold, thus some reductions in the encoding of $\bf/g$ may not be defined.
Evaluation points relative to undefined reductions are called \textit{poles}. 

We find two approaches in the literature to deal with poles: in~\cite{kaltofen2020hermite} an extra symbol $\infty$ is used, 
while in~\cite{guerrini2023simultaneous} coordinates are given by shifted Laurent series representations of the vector of rational functions.

In particular, the authors of~\cite{guerrini2023simultaneous} introduced the following multi-precision encoding composed of a valuation part and a reduction part: 
\begin{defi}[Multi-precision encoding]
	\label{def:MultiprecisionEncoding}
	Given a sequence of evaluation points $\alpha_1,\ldots,\alpha_n\in\F_q[x]$ along with associated multiplicities $\lambda_1,\ldots,\lambda_n\in \Zpos$, and a reduced vector of fractions $\bf/g\in\F_q(x)^\ell$,  we define its multi-precision encoding to be the sequence of couples $\Evi \left(\bf/g\right)\coloneq\left(\nu_{\alpha_j}(g),\cS_j(\bf/g)\right)_{1 \le j \le n}$ such that
	\begin{equation*}
		\cS_j(\bf/g)
		\coloneq
		\bf/\left(g/(x - \alpha_j)^{\nu_{\alpha_j}(g)}\right) \ 
		\bmod
		(x - \alpha_j)^{\lambda_j - \nu_{\alpha_j}(g)}.
	\end{equation*}
	By convention, we set $\cS_j(\bf/g) = \Bold{1}$ when $\nu_{\alpha_j}(g) = \lambda_j$.
\end{defi}

Under the hypothesis $L\ge d_f + d_g - 1$ the authors of~\cite{guerrini2023simultaneous} proved the injectivity of the above encoding.

\begin{prop}
	\label{prop:InjectiveEncodingpoles}
	Let $\bf/g,\bf' /g'\in\F_q(x)^{\ell}$ be reduced with $\max_i\{\partial(f_i)\}< d_f, \ \partial(g) <  d_g$ such that $\Evi \left(\bf/g\right) = \Evi \left(\bf'/g'\right)$. If we assume that $L\ge d_f + d_g - 1$, 
	the equality $\bf/g = \bf' /g'$ holds.
\end{prop}

\begin{proof}
	For every $j=1,\ldots,n$ we let $\vr_j \coloneq\nu_{\alpha_j}(g) = \nu_{\alpha_j}(g')$. By hypothesis $\cS_j(\bf/g) = \cS_j(\bf'/g')$, i.e.
	$
	\bf/\left(g/(x - \alpha_j)^{\vr_j}\right)
	=
	\bf' /\left(g'/(x - \alpha_j)^{\vr_j}\right) \bmod (x - \alpha_j)^{\lambda_j - \vr_j}
	$, thus $
	\bf g'/(x - \alpha_j)^{\vr_j}
	=
	\bf' g/(x - \alpha_j)^{\vr_j} \bmod (x - \alpha_j)^{\lambda_j - \vr_j}.
	$
	In other words $\bf g' = \bf' g  \
	\bmod (x - \alpha_j)^{\lambda_j}$, which implies that $\bf g' = \bf' g \bmod M$. Since by hypothesis $\partial(\bf g' - \bf' g)\le d_f + d_g - 2 < L$, we conclude that $\bf g' = \bf' g$ in $\F_q(x)^{\ell}$.
\end{proof}

Under the hypothesis $L\ge d_f + d_g - 1$, we can then introduce the \textit{simultaneous rational function code with poles} as the set 
\begin{equation*}
	\SRF^{\infty}_{\ell}(M;d_f,d_g)
	\coloneq
	\left\{
	\Evi\left(\frac{\bf}{g}\right)
	:\
	\begin{array}{c}
		\partial(\bf)<d_f, \quad \partial(g)<d_g,\\
		\gcd(f_1,\ldots,f_\ell,g)=1\\
	\end{array}
	\right\}.
\end{equation*}
We will refer to it as the SRF code with poles.

Being composed of two parts, codewords $\Evi(\bf/g)$ can be affected by two kinds of errors (valuation and evaluation errors). 
Here we adapt the hybrid analysis of Section~\ref{sec:Hybrid Decoding}, with the factorization of the error locator $\Lambda = \Li\Lu$ reflecting these two types of errors (see Definition~\ref{def:errorsupportbad primes}).

\begin{defi}
	\label{def:ambientspecebadprimes}
	Let the \textit{ambient space of received words} be the quotient
	\begin{equation*}
		\SL
		\coloneq
		\left(\prod_{j=1}^{n}
		[0,\lambda_j]\times\left(
		\F_q[x]/
		(x - \alpha_j)^{\lambda_j}
		\right)^\ell\right)/\sim
	\end{equation*}
	where $\sim$ is the equivalence relation for which $(\vr_j,\br_j)_{1 \le j \le
		n}\ \sim\ (\vr'_j,\br'_j)_{1 \le j \le n}$ if and only if  for every
	$j=1,\ldots,n, \ (x - \alpha_j)^{\vr'_j}\br_j = (x - \alpha_j)^{\vr_j}\br'_j\ \bmod
	(x - \alpha_j)^{\lambda_j}$. We say that a representative $(\vr_j,\br_j)_{1 \le j \le
		n}$ is \emph{reduced} if $\gcd(\br_j, (x - \alpha_j)^{\vr_j}) = 1$ for every
	$j=1,\ldots,n$. Define $R_i \coloneq \CRT_M(r_{i,1}, \dots, r_{i,n})$ for
	every $i=1,\ldots,\ell$.
\end{defi}

In what follows we can
always assume that the received word $(\vr_j,\br_j)_{1 \le j \le n}$ is reduced, thanks to the following proposition:
\begin{prop}
	Any equivalence class contains a reduced representative.
\end{prop}

\begin{proof}
	Given any received word $(\vr_j,\br_j)_{1 \le j \le n}$,  for
	every $j=1,\ldots,n$ we let
	$(x - \alpha_j)^{\eta_j} \coloneq \gcd(\br_j,(x - \alpha_j)^{\vr_j}) $. Then $$(\vr_j,\br_j)_{1 \le j \le n}\sim\left(\vr_j -
	\eta_j, \br_j^{\lambda_j}/(x - \alpha_j)^{\eta_j} \ \bmod
	(x - \alpha_j)^{\lambda_j}\right)_{1 \le j \le n},$$ with the representative on the right-hand side
	clearly reduced by the definition of $(x - \alpha_j)^{\eta_j}$. 
\end{proof}

In the ambient space $\SL$ we identify received words which represent the same reduced vector of fractions in the sense that, by definition
\begin{itemize}
	\item $(\vr_j,\br_j)_{1 \le j \le n}\
	\sim\
	(\vr_j,\br'_j)_{1 \le j \le n}
	\ \Leftrightarrow
	\br_j = \br'_j \ \bmod (x - \alpha_j)^{\lambda_j - \vr_j}$.
	\item Given a received valuation $0\leq\vr_j\leq\lambda_j$ then for every $1\leq\delta_j\leq \lambda_j - \vr_j$
	\begin{equation*}
		(\vr_j,\br_j)_{1 \le j \le n}\
		\sim\
		(\vr_j + \delta_j,(x - \alpha_j)^{\delta_j}\br_j)_{1 \le j \le n}.
	\end{equation*}
\end{itemize}

\begin{remark}
	Thanks to the first of the above two points we can map the evaluation of a reduced vector of rationals $\Ev^\infty(\bf/g)$ into the space of received words.
\end{remark}

\begin{defi}
	\label{def:DistFuncbad primes}
	Given two elements $\bR_1 := (\vr_j,\br_j)_{1 \le j \le n}, \bR_2 :=
	(\vr'_j,\br'_j)_{1 \le j \le n}$ in $\SL$, we
	define the columns $\be_j$ of the relative error matrix $\bE_{\bR_1,\bR_2}$ as
	\begin{equation*}
		\be_j \coloneq
		(x - \alpha_j)^{\vr_j} \br'_j - (x - \alpha_j)^{\vr'_j} \br_j
		\bmod (x - \alpha_j)^{\lambda_j}.
	\end{equation*}   
	We let the relative error and truth locator be
	\begin{equation*}
		\Lambda_{\bR_1,\bR_2} \coloneq \prod_{j=1}^n(x - \alpha_j)^{\lambda_j - \nu_{\alpha_j}(\be_j)}, 
		\quad
		Y_{\bR_1,\bR_2} \coloneq \prod_{j=1}^n (x - \alpha_j)^{\nu_{\alpha_j}(\be_j)}
	\end{equation*}
	respectively, and the relative distance 
	$d\left(\bR_1,\bR_2\right) \coloneq \partial\left(\Lambda_{\bR_1,\bR_2} \right)$.
\end{defi}

\begin{remark}
	\label{rmk:errorsbad primes}
	Unlike the errors considered in Sections~\ref{Sec:SRF}
	and~\ref{sec:Hybrid Decoding}, in this case the usual relation
	$\bR_1 =\bR_2 + \bE$ does not hold.  
	For this reason the error models (see Subsection~\ref{subsect:Errormodels poles RF}) will be defined directly by distributions in the space of received words $\SL$.
\end{remark}    
In spite of the above remark, we note the consistency of the error $\be_j$ with the equivalence relation $\sim$, indeed by definition 
\begin{equation*}
	\be_j = \boldsymbol{0} \bmod (x - \alpha_j)^{\lambda_j} \quad 
	\forall j = 1,\ldots,n
	\ \ 
	\Leftrightarrow
	\ \
	(\vr_j,\br_j)_{1\le j \le n} \sim (\vr'_j,\br'_j)_{1\le j \le n}.
\end{equation*}

Due to the properties of $\sim$, we can partition the set of error positions into valuation and evaluation errors.
\begin{defi}
	\label{def:errorsupportbad primes}
	Given two evaluations 
	$
	(\vr_j,\br_j)_{1 \le j \le n}$, $(\vr'_j,\br'_j)_{1 \le j \le n} 
	\in \SL
	$ 
	satisfying
	$\gcd((x - \alpha_j)^{\vr_j},\br_j)=1$, 
	we divide the error support
	\begin{equation*}
		\xi=\{j \ | \ (x - \alpha_j)^{\vr_j} \br'_j \neq (x - \alpha_j)^{\vr'_j} \br_j \bmod (x - \alpha_j)^{\lambda_j}\}
		=
		\{j \ | \ (\vr_j,\br_j) \not\sim (\vr'_j,\br'_j)\}
	\end{equation*}
	into the \emph{valuation errors}
	$$\xi_v := \{j \ | \ \vr_j \neq \vr'_j \}$$
	and the remaining \emph{evaluation errors}
	$$\xi_e = \{j \ | \ (\vr_j = \vr'_j) \text{ and } (\br_j \neq \br'_j \bmod
	(x - \alpha_j)^{\lambda_j - \vr_j}) \}.$$
\end{defi}

We provide an equivalent, yet more practical, representation of the errors.

\begin{remark}
	\label{rem:rewriteErrVectorsbad primes}
	Given a codeword $\left(\nu_{\alpha_j}(g),\cS_j(\bf/g)\right)_{1 \le j \le n}$ (as in
	Definition~\ref{def:MultiprecisionEncoding}) and a received word $(\vr_j,\br_j)_{1 \le j \le n}\in
	\SL$,
	the sequence of error vectors $(\be_j)_{1 \le j \le n}$ is given by
	\begin{equation*}
		\be_j = (x - \alpha_j)^{\vr_j} \cS_j(\bf/g) - (x - \alpha_j)^{\nu_{\alpha_j}(g)}\br_j \ \bmod (x - \alpha_j)^{\lambda_j}.
	\end{equation*}
	Multiplying the above by the invertible element $g/(x - \alpha_j)^{\nu_{\alpha_j}(g)}$, we obtain that up to invertible transformations of the error
	sequence components (leaving the distance unchanged), we can equivalently view the sequence of
	error vectors as given by
	\begin{equation*}
		\widetilde{\be}_j \coloneq \frac{g}{(x - \alpha_j)^{\nu_{\alpha_j}(g)}} \be_j =(x - \alpha_j)^{\vr_j}\bf - g\br_j \ \bmod (x - \alpha_j)^{\lambda_j}.
	\end{equation*}
	
\end{remark}

\paragraph*{Study of potential errors and received words around a fixed codeword}
Due to Remark~\ref{rmk:errorsbad primes}, we need to study what kind of errors and received words we can obtain around a fixed vector of fractions $\bf/g$,
in particular with respect to the distinction between valuation and evaluation errors.
Regarding the error positions as long as $\xi_e, \xi_v \subset \{1, \dots, n\}$ and $\xi_e \cap \xi_v = \varnothing$ we have no constraints: all valuation (resp. evaluation) error supports $\xi_v$ (resp. $\xi_e$) are attained.
Once the error positions have been fixed and partitioned as $\xi_v \cup \xi_e$, the valuations of the error vectors need to satisfy $\mu_j = \nu_{\alpha_j}(\be_j) = \lambda_j$ for every position $j$ which is not erroneous, \ie $\forall j \notin \xi_e \cup \xi_v$.
Let us examine what can happen in the evaluation and valuation error cases respectively:
\begin{itemize}
	\item If $j\in \xi_e$, we have an evaluation error, thus any received word $\bR$ must satisfy $\vr_j = \nu_{\alpha_j}(g)$, furthermore we must have that the valuation of any error vector $\be_j$ must satisfy $\mu_j = \nu_{\alpha_j}((x - \alpha_j)^{\vr_j}\bf - g\br_j) \ge \nu_{\alpha_j}(g)$ thus, dividing by $(x - \alpha_j)^{\vr_j}$, we have that $\cS_j(\bf/g) - r_j$ can be any element of valuation 
	$\mu_j - \nu_{\alpha_j}(g)$.
	\item If $j\in\xi_v$, we have a valuation error, thus for every received word we have either 
	\begin{enumerate}
		\item $\vr_j<\nu_{\alpha_j}(g)$: in this case the valuation of the error vector and the received word must coincide, \ie $\mu_j = \vr_j$, and from the definition of $\widetilde{\be}_j$ we must have that $\widetilde{\be}_j = (x - \alpha_j)^{\mu_j}\bf \ \bmod (x - \alpha_j)^{\nu_{\alpha_j}(g)}$, regardless of the reduction part $\br_j$. Thus, in this case we do not have any constraints on $\br_j$.
		\item $\vr_j>\nu_{\alpha_j}(g)$: in this case the valuation of the error vector must coincide with the valuation of $g$, \ie $\mu_j = \nu_{\alpha_j}(g)$. Besides this valuation constraint, the error vectors can take any value, as well as the received reductions $\br_j$.
	\end{enumerate}
\end{itemize}

\paragraph*{Minimal distance}
As for the SRF code without poles (see Lemma~\ref{lm:minDistSRFcode}), also in this case we can prove the following lower bound on the minimal distance of the code.
\begin{lemma}
	\label{lm: MinDistSRFpoles} 
	We have $d\left(\SRF^{\infty}_\ell(M;d_f,d_g)\right) \geq L - d_f - d_g + 2$.
\end{lemma}
\begin{proof}
	Let $\bC_1 = (\nu_{\alpha_j}(g),\cS_j(\bf/g))_{1\le j \le n}, \bC_2 = (\nu_{\alpha_j}(g'),\cS_j(\bf'/g'))_{1\le j \le n}$ be two distinct codewords. 
	From 
	\begin{equation*}
		\be_j 
		=
		(x - \alpha_j)^{\nu_{\alpha_j}(g)}\left(\frac{\bf'}{g'/(x - \alpha_j)^{\nu_{\alpha_j}(g')}}\right)
		-
		(x - \alpha_j)^{\nu_{\alpha_j}(g')}\left(\frac{\bf}{g/(x - \alpha_j)^{\nu_{\alpha_j}(g)}}\right)
		\
		\bmod
		(x - \alpha_j)^{\lambda_j},
	\end{equation*}
	we see that
	\begin{equation*}
		\frac{g}{(x - \alpha_j)^{\nu_{\alpha_j}(g)}}\frac{g'}{(x - \alpha_j)^{\nu_{\alpha_j}(g')}}\be_j =  \bf' g - \bf g' \ \bmod (x - \alpha_j)^{\lambda_j}. 
	\end{equation*}
	Using $\Lambda \be_j = 0 \bmod (x - \alpha_j)^{\lambda_j}$ for all $j$, we obtain
	$$
	\forall 1 \le j \le n, \quad
	0 = \Lambda \frac{g}{(x - \alpha_j)^{\nu_{\alpha_j}(g)}}\frac{g'}{(x - \alpha_j)^{\nu_{\alpha_j}(g')}}\be_j = \Lambda (\bf' g - \bf g') \ \bmod (x - \alpha_j)^{\lambda_j}.
	$$
	Therefore, $M$ divides $\Lambda (\bf' g - \bf g')$, so $Y = M/\Lambda$ divides $(\bf' g - \bf g')$. 
	Hence, for all codewords $\bC_1 \neq \bC_2$, we bound $d(\bC_1,\bC_2) = \partial(\Lambda) = \partial(M/Y) > L - d_f - d_g + 1$, and we have proven the lemma.
    \end{proof}

Due to the different multiplicities among the evaluation points, we need to assume Constraint~\ref{cst:SubsetSumThresholdPoles} on the parameters of the code in order to prove that the lower bound $L - d_f - d_g + 2$ on the minimial distance of the code can be tight (see Lemma~\ref{lm: MinDistSRFpolesTight}).
\begin{constraint}
	\label{cst:SubsetSumThresholdPoles}
	Given the multiplicities $\left(\lambda_1,\ldots,\lambda_n\right)\in \Zpos^n$ and the degree bounds $d_f,d_g\in\Zpos$ such that $L\coloneq \sum_{j=1}^{n}\lambda_j>d_f + d_g - 2$, we assume that there exist two disjoint subsets $S_0,S_{\infty} \subseteq\{1,\ldots,n\}$ and two (not necessarily distinct) indexes $\eta,\gamma \notin S_0 \cup S_{\infty}$, such that:
	\begin{itemize}
		\item $d_f - 1 = \delta_0 + \sum_{j\in S_0}\lambda_j$ with $0 \leq \delta_0 < \lambda_\eta$
		\item $d_g - 1 = \delta_{\infty} + \sum_{j\in S_{\infty}}\lambda_j$ with $0 \leq \delta_{\infty} < \lambda_\gamma$
		\item $\delta_0, \delta_{\infty} > 0 \Rightarrow \eta \ne \gamma$
	\end{itemize}
\end{constraint}

The verification of this constraint represents a particular case of the \textit{multiple subset sum} problem, which is known to be NP-complete on general instances~\cite{caprara2000multiple}).
The proof of the next lemma is going to highlight the reasons behind the conditions of Constraint~\ref{cst:SubsetSumThresholdPoles}.
\begin{lemma}
	\label{lm: MinDistSRFpolesTight} 
    Assuming Constraint~\ref{cst:SubsetSumThresholdPoles}, we have $d\left(\SRF^{\infty}_\ell(M;d_f,d_g)\right) = L - d_f - d_g + 2$.
\end{lemma}
\begin{proof}
	We already proved that the minimal distance is lower bounded by the quantity: $L - d_f - d_g + 2$.
	Thanks to Constraint~\ref{cst:SubsetSumThresholdPoles} we can define the relatively prime polynomials
	\begin{align*}
		f_1 \coloneq \left(x - \alpha_\eta\right)^{\delta_0}\prod_{j\in S_0}(x - \alpha_j)^{\lambda_j}, \quad
		g \coloneq \left(x - \alpha_\gamma\right)^{\delta_{\infty}}\prod_{j\in S_{\infty}}(x - \alpha_j)^{\lambda_j},
	\end{align*}
    whose degrees are respectively $d_f - 1$ and $d_g - 1$ thanks to the first two points of Constraint~\ref{cst:SubsetSumThresholdPoles}. The third point of Constraint~\ref{cst:SubsetSumThresholdPoles} (together with $S_0 \cap S_{\infty} = \emptyset$) ensure that $\gcd(f_1,g)= 1$.
	
	For $\beta\in\F_q\setminus\{0,1\}$, we let  $f_2 \coloneq \beta f_1$ and we define the two codewords $\bC_i \coloneq \Evi\left(f_i\Bold{1}/g\right) = \left(\nu_{\alpha_j}(g),C_{i,j}\right)_{1\leq j \leq n} $ for $i=1,2$, where $\Bold{1}=\left(1,\ldots,1\right)\in\F_q^\ell$.
    We note that for every $j\in S_0\cup S_{\infty}$ the following holds 
    \begin{equation*}
        (\nu_{\alpha_j}(g),C_{1,j})
        =(\nu_{\alpha_j}(g),C_{2,j})
        =
        \begin{cases}
            (0,\Bold{0}) & \text{ if } j \in S_0\\
            (\lambda_j,\Bold{1}) & \text{ if } j \in S_{\infty}
        \end{cases}
    \end{equation*}
    thus $\nu_{\alpha_j}(Y_{\bC_1,\bC_2}) = \lambda_j$ for every $j\in S_0 \cup S_{\infty}$.
    
    Furthermore we see that $\nu_{\alpha_{\eta}}(g) = 0$ and $C_{1,\eta} = C_{2,\eta} = \Bold{0} \bmod (x - \alpha_\eta)^{\delta_0}$, from which it follows that $\nu_{\alpha_{\eta}}(Y_{\bC_1,\bC_2}) \geq \delta_0$. Similarly, we have that $\nu_{\alpha_{\gamma}}\left(Y_{\bC_1,\bC_2}\right) \ge \delta_{\infty}$.
    
	We notice that 
    \begin{align*}
        &\frac{\be_{\eta}}{(x - \alpha_{\eta})^{\delta_0}}
        =
        \frac{\prod_{j\in S_0}(x - \alpha_j)^{\lambda_j}}{\left(x - \alpha_\gamma\right)^{\delta_{\infty}}\prod_{j\in S_{\infty}}(x - \alpha_j)^{\lambda_j}}
        (1 - \beta)\ne \Bold{0}
        \ \bmod \left(x - \alpha_\eta\right)^{\lambda_\eta - \delta_0}\\
        &\frac{\be_{\gamma}}{(x - \alpha_{\gamma})^{\delta_\infty}}
        =
        \frac{\left(x - \alpha_\eta\right)^{\delta_0}\prod_{j\in S_0}(x - \alpha_j)^{\lambda_j}}{\prod_{j\in S_{\infty}}(x - \alpha_j)^{\lambda_j}}
        (1 - \beta)\ne \Bold{0}
        \ \bmod \left(x - \alpha_\eta\right)^{\lambda_\gamma - \delta_{\infty}}
    \end{align*}
    from which it follows that $\bC_1\ne \bC_2$.                            
    
    Thus
	\begin{equation*}
		d_f + d_g - 2 = \left(\sum_{j\in S_0 \cup S_{\infty}}
		\lambda_j\right) + \delta_0 + \delta_{\infty} \leq \partial(Y_{\bC_1,\bC_2}),
	\end{equation*}
	and we conclude that $d\left(\bC_1,\bC_2\right) = L - \partial(Y_{\bC_1,\bC_2}) \leq L - d_f - d_g + 2$
\end{proof}

\begin{remark}
	Constraint~\ref{cst:SubsetSumThresholdPoles} takes inspiration from the proof of Lemma~\ref{lm: MinDistSRFpolesTight} in absence of multiplicities ($\lambda_j = 1$ for every $j=1,\ldots,n$)~\cite[Theorem 2.3.1]{pernet2014high}, in which case it  is a natural consequence of the hypothesis $L>d_f + d_g - 2$, by taking $S_0 = \{1,2,\ldots,d_f - 1\}, S_{\infty} = \{d_f + 1, d_f + 2, \ldots, d_f + d_g - 1\}$, $\delta_0 = \delta_{\infty} = 0$ and $\eta = \gamma = d_f$.
\end{remark}

\begin{coro}
	Given the parameters of the code $\Bold{\lambda} \coloneq \left(\lambda_1,\ldots,\lambda_n\right)$, $d_f,d_g$ satisfying the injectivity condition of Proposition~\ref{prop:InjectiveEncodingpoles}, \ie $L \geq d_f + d_g - 1$. If $\Bold{\lambda} = \lambda\Bold{1}$ for some $0<\lambda<\min\{d_f,d_g\}$, then Constraint~\ref{cst:SubsetSumThresholdPoles} is satisfied.
\end{coro}
\begin{proof}
	Since $0<\lambda<\min\{d_f,d_g\}$, we can write by Euclidean divisions
	\begin{align*}
		d_f - 1 = q_0 \lambda + \delta_0,& \quad 0\leq\delta_0<\lambda\\
		d_g - 1 = q_{\infty} \lambda + \delta_{\infty},& \quad 0\leq\delta_{\infty}<\lambda.
	\end{align*}
	Substituting the above in the hypothesis $L = \lambda n \geq d_f + d_g - 1$ we obtain:
	\begin{equation*}
		n \geq q_0 + q_{\infty} + \frac{\delta_0 + \delta_{\infty} + 1}{\lambda}> q_0 + q_{\infty},
	\end{equation*} 
	thus we can let $S_0, S_{\infty}$ be any two disjoint subsets with $\# S_0 = q_0, \# S_{\infty} = q_{\infty}$ and Constraint~\ref{cst:SubsetSumThresholdPoles} holds.  
\end{proof}

\subsection{Key equations}
As for the code of Section~\ref{Sec:SRF}, the decoding of SRF codes with poles is based on the resolution of a system of linear \textit{key equations}.
Thanks to Remark~\ref{rem:rewriteErrVectorsbad primes} and the definition of $\Lambda$, we have that 
$\Lambda \be_j = 0 \bmod (x - \alpha_j)^{\lambda_j}$, and so 
$0 = \Lambda \widetilde{\be}_j = (x - \alpha_j)^{\vr_j} \Lambda \bf - \Lambda g\br_j \bmod (x - \alpha_j)^{\lambda_j}$.
Thus, for every couple of received word $\left(\vr_j,\br_j\right)_{1 \le j \le n}$ and reduced vector of
fractions $\bf/g$ the equation 
$
\CRT_M\left((x - \alpha_j)^{\vr_j}\right) \Lambda f_i
=
\Lambda g R_i \ \bmod M
$
holds for every $i=1,\ldots,\ell$. By defining the new variables $\varphi \coloneq \Lambda g$,
$\Bold{\psi} = \Lambda \bf$ we get the \textit{key equations} in presence of poles: 
\begin{equation}
	\label{keyEqRF poles}
	\forall i=1,\ldots,\ell, 
	\quad
	\CRT_M\left((x - \alpha_j)^{\vr_j}\right) \psi_i
	=
	\varphi R_i
	\
	\bmod M.
\end{equation} 
For some distance parameter $t$, we let the set of solutions be
\begin{align*}
	\SRinf
	\coloneq
	\left\{\left(\varphi,\Bold{\psi}\right)\in\F_q[x]^{\ell+1}:
	\begin{array}{c}
		\CRT_M\left((x - \alpha_j)^{\vr_j}\right) \psi_i
		=
		\varphi R_i
		\
		\bmod M,
		\quad\forall i
		\\
		\partial(\varphi) < d_g + t,
		\
		\partial(\Bold{\psi}) < d_f + t
	\end{array}
	\right\}.
\end{align*} 

If $\partial(\Lambda)\leq t$ we see that $v_{\bC}\coloneq(\Lambda g,\Lambda\bf)\in\SRinf$.

\paragraph*{Reduced key equations}
It is possible to give an equivalent description of the solutions in $\SRinf$, whose degree bounds are smaller.
Letting $M_{\infty}:=\prod_{j=1}^n (x - \alpha_j)^{\vr_j}$ we note that, thanks to Equation~\eqref{keyEqRF poles}, $M_{\infty}|\varphi$ since $M_{\infty}|\CRT_M\left((x - \alpha_j)^{\vr_j}\right)$, $M_{\infty}|M$ and by hypothesis $\gcd(M_{\infty},R_i) = 1$ as received words are assumed to be reduced. Thus, we can rewrite Equation~\eqref{keyEqRF poles} in the following form, which we call \textit{reduced key equations}
\begin{equation}
	\label{keyEqRF polesReduit}
	\forall i=1,\ldots,\ell, 
	\quad
	\begin{array}{c}
		\psi_i
		=
		\varphi' R_i' \bmod \frac{M}{M_{\infty}} \\
		\partial(\varphi')< d_g + t - \partial(M_{\infty}) \text{ and } \partial(\Bold{\psi})< d_f + t
	\end{array}
\end{equation}
where $\varphi' \coloneq \varphi/M_{\infty}$ and $R_i' \coloneq R_i \CRT_{M/M_{\infty}}\left(\frac{M_{\infty}}{(x - \alpha_j)^{\vr_j}}\right)$.

\subsection[Decoding SRF codes with poles]{Decoding $\SRF^\infty_\ell$ codes}
\label{subsect:DecodingSRNcodesbad primes}
In this section we give our decoding algorithm for $\SRF$ codes with poles, which performs the same steps as Algorithm~\ref{algoSRF} only on a different system of linear equations. 

\begin{algorithm}[H]
	\caption{$\SRF^\infty_\ell$ codes decoder.}
	\label{algoSRFpoles}
	\SetAlgoLined \SetKw{KwBy}{par} \KwIn{$\SRF^{\infty}_\ell(M;d_f,d_g)$, received word $\bR:=(\vr_j,\br_j)_{1\le j\le n}$, distance
		bound $t$} \KwOut{A reduced vector of fractions $\bPsi'/\varphi'$ s.t. $d(\Evi(\bPsi'/\varphi'),\bR) \leq t$ or ``decoding failure''}
	
	\vspace{5pt}

    Compute $\Bold{0}\ne v_s \coloneq \left(\varphi,\psi_1,\ldots,\psi_{\ell}\right)\in\SRinf$, s.t. $\max\{\partial(\varphi),\partial(\Bold{\psi})\}$ is minimal.\\
	Let $\eta \coloneq \gcd(\varphi,\psi_{1},\dots,\psi_{\ell})$, $\varphi'\coloneq\varphi/\eta$ and
	$\forall i, \ \psi'_i\coloneq\psi_i/\eta$ \label{step:lambdabad primes}\\
	\If{$\partial(\eta) \le t$, $\partial(\varphi') < d_g$ and $\forall i, \ \partial(\psi_i')< d_f$}
	{\textbf{return} $(\psi_1'/\varphi', \dots, \psi_\ell'/\varphi')$}
	\lElse{ \textbf{return} "decoding failure"}
\end{algorithm}
\begin{lemma}
	\label{lm:succeedsImpliesCorrectbad primes} 
	If Algorithm~\ref{algoSRFpoles} returns $\bPsi'/\varphi'$ on input $\bR$ and
	parameter $t$, then $\bPsi'/\varphi'$ is associated to a code word of
	$\SRF^{\infty}_{\ell}(M;d_f,d_g)$ close to $\bR$, \ie it is a reduced vector of fractions with
	$\partial(\bPsi')<d_f$, $\partial(\varphi')<d_g$ and
	\mbox{$\dist(\Ev^\infty(\bPsi'/\varphi'),\bR) \le t$}.
\end{lemma}

\begin{proof}
	The output vector $\bPsi/\varphi$ is associated to a code word of
	$\SRF^{\infty}_{\ell}(M;d_f,d_g)$ since the algorithm has verified the degree conditions $\partial(\varphi')<d_g$,
	$\partial(\bPsi')<d_f$.
	Now, we use that $(\varphi,\bPsi)=(\eta \varphi',\eta\bPsi')\in\SRinf$, so that for every $i =1,\ldots,\ell$ the equation $
    \eta\left(\CRT_M((x - \alpha_j)^{\vr_j}) \bPsi' - \varphi' R_i\right) = 0 \bmod M
    $ holds,
     which implies that $\nu_{\alpha_j}(\eta)\geq \lambda_j - \mu_j =
	\nu_{\alpha_j}(\Lambda)$ with $\Lambda$ being the error locator between $\Evi (\bPsi/\varphi)$ and the input $\bR$. Thus, $\Lambda|\eta$, and we can conclude that
	$d(\Ev^\infty(\bPsi'/\varphi'),\bR) = \partial(\Lambda) \le \partial(\eta) \le t$.
\end{proof}

Also in this setting (with the same proof idea), we can state an equivalent of Lemma~\ref{lm:algofailureconditionPOLY}, ensuring that for the random distributions specified in Subsection~\ref{subsect:Errormodels poles RF}, the failure probability of Algorithm~\ref{algoSRFpoles} can be upper bounded by $\P(\SRinf\not\subseteq v_{\bC}\F_q[x])$.  

\subsection{Unique decoding}

We notice that Algorithm~\ref{algoSRFpoles} it is always correct in decoding SRF codes with poles whenever the distance parameter $t\le \lf\frac{L - d_f - d_g + 1}{2} \rf$, \ie it is below unique decoding capacity.

\begin{prop}
	\label{Unicity}
    Given the SRF code with poles and a received word $\bR = (\vr_j,\br_j)_{1 \le j \le n}$, we suppose that the distance parameter input of Algorithm~\ref{algoSRFpoles} satisfies $t \leq  \lf\frac{L - d_f - d_g + 1}{2} \rf$. Furthermore, suppose there exists a reduced vector of rational functions $\bf/g\in\F_q(x)^{\ell}$ with $\partial(\bf)< d_f$ and $\partial(g)< d_g$ such that $d(\Evi(\bf/g), (\vr_j,\br_j)_{1 \le j \le n}) \le t$, then $\SRinf \subset v_{\bC} \F_q[x]$.
\end{prop}
\begin{proof}
	By hypothesis $d(\Evi(\bf/g), (\vr_j,\br_j)_{1 \le j \le n}) \le t$, we have $v_{\bC} = (\Lambda g,\Lambda\bf)\in\SRinf$.
	Let $(\varphi,\bPsi)\in\SRinf$ be another solution of the key equations.
	We have that 
	\begin{equation*}
		\begin{cases}
			(x - \alpha_j)^{\vr_j}\Lambda\bf =
			\br_j \Lambda g &\phantom{}\bmod (x - \alpha_j)^{\lambda_j}\\
			(x - \alpha_j)^{\vr_j}\bPsi\ =
			\br_j \varphi &\phantom{}\bmod (x - \alpha_j)^{\lambda_j}
		\end{cases}
	\end{equation*}
	for some 
	$\Lambda\in\F_q[x]$ with $\partial\left(\Lambda\right) \le t
	\le\lf\frac{L - d_f - d_g + 1}{2} \rf$. Since the received word $(\vr_j,\br_j)_{1 \le j \le n}$ is assumed to be reduced, \ie $\gcd((x - \alpha_j)^{\vr_j},\br_j) = 1$, from the above we get that 
	$(x - \alpha_j)^{\vr_j}|\Lambda g$ and $(x - \alpha_j)^{\vr_j}|\varphi$ for every $j=1,\ldots,n$. Thus, when multiplying the first equation by $\varphi$ and the second one by
	$\Lambda g$ we get
	\begin{align*}
		\begin{cases}
			(x - \alpha_j)^{\vr_j}\Lambda\varphi \bf =
			\br_j \Lambda g \varphi &\phantom{}\bmod (x - \alpha_j)^{\lambda_j + \vr_j}\\
			(x - \alpha_j)^{\vr_j}\Lambda g\bPsi  =
			\br_j \Lambda g \varphi &\phantom{}\bmod (x - \alpha_j)^{\lambda_j + \vr_j}
		\end{cases}
	\end{align*}
	subtracting one another, and dividing by $(x - \alpha_j)^{\vr_j}$, we obtain 
	$
	\Lambda \left(\varphi\bf - g\bPsi \right) = 0 \bmod M.
	$
	By hypothesis, we have that
	$
	\partial\left(\Lambda (\varphi\bf - g\bPsi)\right) \leq 
	2t + d_f + d_g - 2 < L 
	$
	which implies that $\Lambda \left(\varphi\bf - g\bPsi \right) = 0$ thus $\varphi\bf
	= g\bPsi $ in $\F_q[x]^\ell$.  
	
	Since $\bf/g$ is a reduced vector of rational functions, there exists $p\in\F_q[x]$ such that $(\varphi,\bPsi) = p (g,\bf)$.
	Substituting in the key equations for $(\varphi,\bPsi)$, we get
	$
	p ((x - \alpha_j)^{\vr_j}\bf - \br_j g) = 0  \bmod (x - \alpha_j)^{\lambda_j}
	$.
	However, $\Lambda$ divides $p$ by definition of $\Lambda$, so $(\varphi,\bPsi) \in v_{\bC} \F_q[x]$.
\end{proof}

\subsection{Hybrid Error Models for Poles}
\label{subsect:Errormodels poles RF}
In this subsection we adapt the hybrid error models of the previous section to the current context with poles. Recall that the hybrid error model is composed of both
fixed errors and random errors. As done in~\cite{guerrini2023simultaneous}, here we consider a hybrid error model where valuation errors are fixed, while evaluation errors are random.
In previous Sections~\ref{Sec:SRF} and~\ref{sec:Hybrid Decoding}, the error models were defined on the error
matrices $\bE$, then the theorems applied to received words $\bR$ such that $\bR
= \bC + \bE$. In this section, as pointed out in Remark~\ref{rmk:errorsbad primes}, we have a more complicated relation between
$\bR$, $\bC$ and $\bE$. So we are going to define the error model directly on
$\bR$.

Our error model needs to fix the following parameters:
\begin{itemize}
	\item a reduced vector of rational functions $\bf/g\in\F_q(x)^{\ell}$ such that
	$
	\partial(\bf)<d_f,$ $ \partial(g)<d_g
	$,
	
	\item valuation $\xi_v$ and evaluation $\xi_e$ error supports such that
	$\xi_e, \xi_v \subset \{1, \dots, n\}$ and $\xi_e \cap \xi_v = \varnothing$,
	
	\item error valuations $(\mu_j)_{1 \le j \le n}$ such that 
	\begin{itemize}
		\item $\mu_j = \lambda_j$ for $j \notin \xi_e \cup \xi_v$,
		\item $\mu_j \ge \nu_{\alpha_j}(g)$ and $\mu_j < \lambda_j$ for $j \in \xi_e$,
		\item $\mu_j \le \nu_{\alpha_j}(g)$ and $\mu_j < \lambda_j$ for $j \in \xi_v$,
	\end{itemize}
	
	\item a partial received word $\bfR_j = (\vr_j, \br_j)$ for all $j \in
	\xi_v$ such that 
	\begin{itemize}
		\item $\vr_j = \mu_j$ when $\mu_j < \nu_{\alpha_j}(g)$,
		\item $\vr_j > \nu_{\alpha_j}(g)$ when $\mu_j = \nu_{\alpha_j}(g)$.
	\end{itemize}
\end{itemize}
Denote $\Le := \prod_{j \in \xi_e} (x - \alpha_j)^{\lambda_j - \mu_j}, \Lv := \prod_{j \in
	\xi_v} (x - \alpha_j)^{\lambda_j - \mu_j}$ and $\Lambda = \Le \Lv$. Remark that $\Le, \Lv,
\Lambda$ contain all the information of $\xi_v, \xi_e$ and $\mu_j$ since $\xi_v
= \cZ(\Lv)$, $\xi_e = \cZ(\Le)$ and $\mu_j = \lambda_j - \nu_{\alpha_j}(\Lambda)$.

We are ready to define our error models. The random received words $\bR = (\vr_j, \br_j)_j$ are uniformly
distributed in the following set $\HYBp$
\begin{enumerate}
	\item $\bR_j = \Evi(\bf/g)_j \text{ for all } j \text{ such that } \alpha_j\not\in\cZ(\Lambda)$,
	\item $\bR_j = \bfR_j \text{ for all } j\text{ such that } \alpha_j\in\cZ(\Lv)$,
	\item $\bR_j = (\nu_{\alpha_j}(g), \br_j)$ with $\nu_{\alpha_j}\left(\br_j -
	\cS_j(\bf/g)\right) = \mu_j - \nu_{\alpha_j}(g)$ for all $j$ such that $\alpha_j\in\cZ(\Le)$.
\end{enumerate}

As before, we will determine the distribution of the error matrices
$\bE_{\bR,\Evi(\bf/g)}$ when $\bf/g$ is fixed and $\bR$ is random. For
$i\in\{1,\dots,\ell\}$, we still denote $E_i \in \F_q[x]/M$ the CRT interpolant of
the $i$-th row of $\bE$, and we obtain that $Y | E_i$ for every index $i=
1,\ldots,\ell$ as in Subsection~\ref{subsect:ErrModelRatFunc}. We define the
modular polynomials $E_i'\coloneq E_i/Y \in \F_q[x]/\Lambda$, which verify
$\gcd(E_1',\ldots,E_{\ell}',\Lambda) = 1$.

Because of our hybrid error model where the randomness only appears on the
columns $j \in\cZ(\Le)$, we need to study the random vector $(E_1' \bmod \Le,\ldots,E_{\ell}' \bmod \Le)$.
\begin{lemma}
	\label{lm:randombad primesuniform}
	If $\bR$ is uniformly distributed in $\HYBp$, then 
	the random vector $(E_1' \bmod \Le,\ldots,E_\ell' \bmod \Le)$ is uniformly
	distributed in the sample space $\Omega_{\Lambda_e,\ell}$.
\end{lemma}
\begin{proof}
	For the duration of this proof, we will only consider indices $j$ such that
	$p_j\in\cZ(\Le)$. Recall that
	$
	\be_{j} = (x - \alpha_j)^{\nu_{\alpha_j}(g)} (\br_{j} - \cS_j(\bf/g)) \bmod (x - \alpha_j)^{\lambda_j}
	$
	for all those particular $j$.
	Since $\nu_{\alpha_j}(Y) = \nu_{\alpha_j}(\be_{j}) = \lambda_j - \nu_{\alpha_j}(\Le)$,
	we get that the vector $\be_{j}/Y \in
	(\F_q[x]/(x - \alpha_j)^{\nu_{\alpha_j}(\Le)})^\ell$ is random of valuation $0$. As a
	consequence, by the definition of $\HYBp$, we obtain that $(E_1' \bmod \Le,\ldots,E_\ell' \bmod \Le)$ is random
	among the vectors of $(\F_q[x]/\Le)^\ell$ such that
	$\gcd(E_1',\ldots,E_{\ell}',\Le) = 1$.
\end{proof}

\paragraph*{Second error model}
Similarly, we need to fix a reduced vector of rational functions $\bf/g\in\F_q(x)^{\ell}$,
valuation $\xi_v$ and evaluation $\xi_{m,e}$ error supports, error valuations
$(\mu_j)_{1 \le j \le n}$ and a partial received word $\bfR_j = (\vr_j, \br_j)$
for all $j \in \xi_v$. All these parameters must satisfy the same conditions as
the first error model. 

The set $\xi_{m,e}$ is now called the maximal error
support because actual errors could result in an evaluation error support
$\xi_{e} \subset \xi_{m,e}$.
Denote $\Lambda_{m,e} := \prod_{j \in \xi_{m,e}} (x - \alpha_j)^{\lambda_j - \mu_j}$,
$\Lv := \prod_{j \in \xi_v} (x - \alpha_j)^{\lambda_j - \mu_j}$ and $\Lambda_m = \Lambda_{m,e} \Lv$.

In the second error model, the random received words $\bR = (\vr_j, \br_j)_j$
are uniformly distributed in the following set $\HYBpd$
\begin{enumerate}
	\item $\bR_j = \Evi(\bf/g)_j \text{ for all } j \text{ such that } \alpha_j\not\in\cZ(\Lambda_m)$,
	\item $\bR_j = \bfR_j \text{ for all } j\text{ such that } \alpha_j\in\cZ(\Lv)$,
	\item $\bR_j = (\nu_{\alpha_j}(g), \br_j)$ with $\nu_{\alpha_j}\left(\br_j -
	\cS_j(\bf/g)\right) \ge \mu_j - \nu_{\alpha_j}(g)$ for all $j$ such that $\alpha_j\in\cZ(\Lambda_{m,e})$.
\end{enumerate}
Notice that for a given received word in the set $\HYBpd$, the associated
error locator has the form $\Lambda = \Le\Lv$ for some divisor $\Le|\Lme$.

\subsection{Our results on poles}

We are ready to state our results regarding the failure probability of the decoding algorithm in
presence of poles. We let $\tme$ be the maximal distance on the evaluation errors 
\begin{equation}
	\label{def:Dmaxpoles}
	\tme \coloneq 
	\frac{\ell}{\ell+1} \left[L - d_f - d_g + 1 - 2 \htv\right]
\end{equation}
\begin{theorem}
	\label{thm:main1poles}
	Decoding Algorithm~\ref{algoSRFpoles} on input 
	\begin{enumerate}
		\item distance parameter $t = \htv
		+ \hte$ for $\htv\leq \lf\frac{L - d_f - d_g + 1}{2} \rf$
		and $\hte\leq \tme$,
		\item a random received word $\bR = (\vr_j,\br_j)_{1 \le j \le n}$
		uniformly distributed in $\HYBp$, for some reduced vector of rational functions
		$\bf/g\in\F_q(x)^{\ell}$ with $\partial(\bf) < d_f, \ \partial(g)<d_g$, and $\partial
		(\Lv) \leq \htv$ and $\partial(\Le)\leq \hte$,
	\end{enumerate}
	outputs the center vector $\bf/g$ of the distribution
	with a probability of failure 
	\begin{equation*}
		\P_{\fail}
		\leq 
		q^{-(\ell+1)(\tme - \hte)}\prod_{\alpha\in\cZ(\Le)}\left(\frac{1 - 1/q^{\ell + \nu_{\alpha}(\Lambda_{e)}}}{1 - 1/q^{\ell}}\right).
	\end{equation*} 
\end{theorem}
\begin{theorem}
	\label{thm:main2poles}
	Decoding Algorithm~\ref{algoSRFpoles} on input 
	\begin{enumerate}
		\item distance parameter $t = \htv
		+ \hte$ for $\htv\leq \lf\frac{L - d_f - d_g + 1}{2} \rf$
		and $\hte\leq \tme$,
		\item a random received word $\bR = (\vr_j,\br_j)_{1 \le j \le n}$
		uniformly distributed in $\HYBpd$, for some reduced vector of rational functions
		$\bf/g\in\F_q(x)^{\ell}$ with $\partial(\bf) < d_f, \ \partial(g)<d_g$, and $\partial
		(\Lv) \leq \htv$ and $\partial(\Lme)\leq \hte$,
	\end{enumerate}
	outputs the center vector $\bf/g$ of the distribution
	with a probability of failure 
	\begin{equation*}
		\P_{\fail}
		\leq 
		q^{-(\ell+1)(\tme - \hte)}
		\prod_{\alpha\in\cZ(\Lme)}
		\left(\frac{1 - 1/q^{\ell + \nu_{\alpha}(\Lme)}}{1 - 1/q^{\ell + 1}}\right).
	\end{equation*} 
\end{theorem}

\begin{remark}
	We remark that the results given in this paper provide several
	improvements on the state of the art 
	(see~\cite[Theorem 3.4]{guerrini2023simultaneous}). For instance, the
	failure probability bound decreases exponentially when the actual error
	distance is less than the maximal error distance in this paper, whereas the failure probability bound in~\cite{guerrini2023simultaneous} is a linear function of the distance parameter. Furthermore, our bound removes the technical dependency of the multiplicity balancing, making the results independent of how the multiplicities are distributed.
	
\end{remark}

\subsection{Decoding failure probability with respect to the first error model} 
We let 
\begin{equation*}
	\SEinf\coloneq
	\left\{
	\omega\in\F_q[x]/\Le : \forall i, \ \omega \widetilde{E}'_i \in \F_q[x]_{\Le,B + \partial(\Lambda)}
	\right\}
\end{equation*}
with $B\coloneq d_f + d_g + t - L - 2 = d_f + d_g + \hte +\htv - L - 2  $ and
$\widetilde{E}_i \coloneq \CRT_M\left(g/(x - \alpha_j)^{\nu_{\alpha_j}(g)}\right) E_i \ \bmod M$. 
We can now prove the version of Lemma~\ref{BaseLemmaRN} with poles.
\begin{constraint}
	\label{cst: Param polesVers}
	The parameters of Algorithm~\ref{algoSRFpoles} satisfy $B + \htv < 0$.
\end{constraint}
\begin{lemma}
	\label{lm:SE_bad primes}
	If Constraint~\ref{cst: Param polesVers} is satisfied then $\SEinf = \{0\} \Rightarrow
    \SRinf	\subseteq v_{\bC}\F_q[x]$.
\end{lemma}
\begin{proof}
	Let $(\varphi,\psi_1,\ldots,\psi_\ell)\in \SRinf$.
	 From~\eqref{keyEqRF poles} we know that $\prod_{j=1}^n(x - \alpha_j)^{\vr_j}|\varphi$ and that for every $i,j$ there exists $h_{i,j}\in\F_q[x]$ such that $\varphi r_{i,j}=(x - \alpha_j)^{\vr_j}\psi_i + h_{i,j} (x - \alpha_j)^{\lambda_j}$. Furthermore,
	\begin{equation}
		\label{eq:intermediatefromLemmaSEbad primes}
		\varphi \Lv \widetilde{e}_{i,j}
		= 
		(x - \alpha_j)^{\vr_j}\Lv\left(\varphi f_i - g\psi_i\right) - \Lv g h_{i,j} (x - \alpha_j)^{\lambda_j} \ \bmod (x - \alpha_j)^{\lambda_j + \vr_j}.    
	\end{equation}
	From
	\begin{equation*}
		\nu_{\alpha_j}(\Lv g)=
		\begin{cases}
			\lambda_j - \min\{\vr_j,\nu_{\alpha_j}(g)\} + \nu_{\alpha_j}(g) & \text{ if }\vr_j\ne\nu_{\alpha_j}(g)\\
			\nu_{\alpha_j}(g)& \text{ if }\vr_j=\nu_{\alpha_j}(g)
		\end{cases},
	\end{equation*}
	as $\lambda_j\geq \vr_j $, we conclude that $\nu_{\alpha_j}(\Lv g)\geq \vr_j$ for every $j=1,\ldots,n$. Taking the CRT interpolant modulo $M$ on both sides of~\eqref{eq:intermediatefromLemmaSEbad primes} after dividing by $(x - \alpha_j)^{\vr_j}$, we conclude that
	$$
	\CRT_M(\varphi/(x - \alpha_j)^{\vr_j})\Lv \widetilde{E}_{i} 
	=
	\Lv
	(
	\varphi f_i - g \psi_i
	)
	\ \bmod M
	$$
	with $\widetilde{E}_i \coloneq \CRT_M\left(g/(x - \alpha_j)^{\nu_{\alpha_j}(g)}\right) E_i \ \bmod M$.
	The polynomial $Y\Lv$ divides both $\Lv \widetilde{E}_{i}$ and $M$, so it divides  
	$\Lv ( \varphi f_i - g \psi_i )$.
	Dividing by $Y\Lv$, we obtain 
	\begin{equation*}
		\CRT_{\Le}\left(\frac{\varphi}{(x - \alpha_j)^{\vr_j}}\right)\widetilde{E}'_{i} 
		=
		\frac{\varphi f_i - g \psi_i}{Y}
		\ \bmod \Le,
	\end{equation*}
	with $\widetilde{E}'_i \coloneq \CRT_{\Le}\left(g/(x - \alpha_j)^{\nu_{\alpha_j}(g)}\right) E'_i \ \bmod
	\Le$. Thus, $\omega \coloneq \CRT_{\Le}\left(\varphi/(x - \alpha_j)^{\vr_j}\right)\in      
	\SEinf$ and, thanks to the hypothesis $\SEinf = \{0\}$, $\left(\varphi
	f_i - g \psi_i\right)/Y = 0 \ \bmod\Le$. Thanks to Constraint~\ref{cst: Param polesVers} and since $\partial(\Lv)\leq \htv$, we have
	$
	\partial\left(\frac{g\psi_i-f_i\varphi}{Y}\right) \leq B + \partial(\Lambda) < \partial(\Le).
	$
	As a result, $g\psi_i = f_i\varphi$ for all $i=1,\ldots,\ell$.
	Since $\gcd(f_1,\ldots,f_\ell, g)=1$,  we must have that $g|\varphi$, i.e. $\varphi = sg$ for
	some $s\in\F_q[x]$ and, from the above conclusion, as well that $\Bold{\psi} = s\bf$. Let us
	note 
	\begin{equation*}
		s\widetilde{\be}_j =
		(x - \alpha_j)^{\vr_j}\Bold{\psi} - \varphi\br_j = \Bold{0} \bmod (x - \alpha_j)^{\lambda_j}.
	\end{equation*} 
	As $\nu_{\alpha_j}\left(\widetilde{\be}_j\right) = \nu_{\alpha_j}\left(\be_j\right) = \lambda_j -
	\Val_j(\Lambda)$, we obtain $\nu_{j}(s)\ge \lambda_j - (\lambda_j - \Val_j(\Lambda))$
	for every $j$, i.e. $\Lambda$ divides $s$.
\end{proof}

\begin{remark}
	Along the same lines of Remark~\ref{rmk:unicity} relative to the analysis of Section~\ref{Sec:SRF}, also in this context we see that when the distance parameter $t$ of the decoding algorithm satisfies $t<\lf\frac{L - d_f - d_g + 1}{2}\rf$ and thus, thanks to Lemma~\ref{lm: MinDistSRFpoles}, it is
	below half of the minimal distance of the SRF code with poles, we must have that $B + \partial(\Lambda) \leq B + t < 0$ since
	$\partial(\Lambda)\le t$.
	Under such circumstance we therefore
	have $\F_q[x]_{\Le, B + \partial(\Lambda)} = \{0\}$ and thus, similarly to Remark~\ref{rmk:unicity}, estimating the
	failure probability of Algorithm~\ref{algoSRFpoles} by studying $\P(\SEinf\ne\{0\})$
	yields the expected unique decoding result whenever $t<\lf\frac{L - d_f - d_g + 1}{2}\rf$.
\end{remark}

\begin{proof}[Proof of Theorem~\ref{thm:main1poles}]
	Since for all received words in
	our random distribution, we know that if $\SRinf\subseteq v_{\bC}\F_q[x]$ then Algorithm~\ref{algoSRFpoles} succeeds, thus by contrapositive  
	$\P_{\fail} \le \P(\SRinf\not\subseteq v_{\bC}\F_q[x])$.
	
	We can prove that our choice of parameters satisfy Constraint~\ref{cst: Param polesVers} in the same fashion as the proof of
	Theorem~\ref{thm:main1hyb}. So we can apply Lemma~\ref{lm:SE_bad primes} to
	obtain $\P(\SRinf\not\subseteq v_{\bC}\F_q[x]) \le \P\left(\SEinf \ne
	\{0\}\right)$.
	
	As in Equation~\eqref{boundInc/ExclSe}, we have 
	$
	\P\left(\SEinf \ne \{0\}\right)
	\leq \frac{1}{q - 1}
	\sum_{\omega=1}^{\Le - 1}
	\P\left(
	\forall i, \  \omega \widetilde{E}_i' \in \F_q[x]_{\Le,B + \partial(\Lambda)}
	\right)
	$.
	Since $\widetilde{E}_i = \CRT_M\left(g/(x - \alpha_j)^{\nu_{\alpha_j}(g)}\right) E_i \ \bmod M$, and since
	$\CRT_M\left(g/(x - \alpha_j)^{\nu_{\alpha_j}(g)}\right)$ is an invertible element of $\F_q[x]/M$, we have that for
	every $1\leq\omega\leq\Le - 1$,\linebreak
	$
	\P(
	\forall i, \  \omega \widetilde{E}_i' \in \F_q[x]_{\Le,B + \partial(\Lambda)}
	) 
	=
	\P(
	\forall i, \  \omega E_i' \in \F_q[x]_{\Le,B + \partial(\Lambda)}
	).
	$
	Now, since we know the distribution of $(E_i')_{1 \le i \le n}$ thanks to
	Lemma~\ref{lm:randombad primesuniform}, we use
	Lemma~\ref{lm:boundfailureprobabilityhybrid} with $\Li$ and $\htu$ being
	replaced by $\Le$ and $\htv$ to get
	\begin{equation*}
		\sum_{\omega=1}^{\Le - 1}
		\P\left(
		\forall i, \  \omega E_i' \in \F_q[x]_{\Le,B + \partial(\Lambda)}
		\right) 
		\leq
		q^{\ell(B +1 +\htv )}
		q^{\partial(\Le)}
		\prod_{\alpha\in\cZ(\Le)}\left(\frac{1 - 1 / q^{\ell + \nu_{\alpha}(\Le)}}{1 -1 / q^{\ell}}
		\right).
	\end{equation*}
	Since $\partial(\Le)\leq \hte$, we have
	$
	q^{\ell(B +1 +\htv )}
	q^{\partial(\Le)}
	\le 
	q^{\ell(B +1 +\htv )}
	q^{\hte}
	=
	q^{-(\ell+1)(\tme - \hte)}
	$, we have proven Theorem~\ref{thm:main1poles}.
\end{proof}

\subsection{Decoding failure probability with respect to the second error model}
\label{subsubsect: proofbad primesRN-D2}

We will denote $\P_{\HYBpd}$ (resp. $\P_{\HYBp}$) the probability function under
the second (resp. first) error model specified by a given factorization $\Lambda_{m,e}\Lv$ of
the error locator, and by a partial received word $(\bfR_j)_{j \in \cZ(\Lv)}$.

\begin{proof}[Proof of Theorem~\ref{thm:main2poles}]
	As done in the proof of Theorem~\ref{thm:main2hyb}, letting $\cF$ be the event of decoding failure and
	using the law of total probability, we have that
	$\P_{\HYBpd}( \cF )$ can be decomposed as the sum
	\begin{equation*}
		\sum_{\Le|\Lme}
		\P_{\HYBpd}(\cF \ |\ \Lambda_{\bE} = \Le\Lv)
		\ 
		\P_{\HYBpd}( \Lambda_{\bE} = \Le\Lv),
	\end{equation*}
	where 
	\begin{equation*}
		\P_{\HYBpd}(\cF \ |\ \Lambda_{\bE} = \Le\Lv)
		=
		\P_{\HYBp}\left(\cF\right)
	\end{equation*}
	is upper bounded by
	\begin{equation*}
		\P_{\HYBp}(\cF)
		\leq
		\frac{q^{\ell(B +1 +\htv )}
			q^{\partial(\Le)}}{q - 1}
		\prod_{\alpha\in\cZ(\Le)}\left(\frac{1 - 1 / q^{\ell + \nu_{\alpha}(\Le)}}{1 -1 / q^{\ell}}
		\right).
	\end{equation*}
	Whereas
	\begin{align*}
		\P_{\HYBpd}(\cF \ |\ \Lambda_{\bE} = \Le\Lv)
		&=
		\frac{\prod_{\alpha\in\cZ(\Le)}(q^{\ell} -1) q^{\ell(\nu_{\alpha}(\Le) - 1)}}{\prod_{\alpha\in\cZ(\Lme)}q^{\ell \nu_{\alpha}(\Lme)}}\\
		&=
		\frac{q^{\ell\partial(\Le)}}{q^{\ell\partial(\Lme)}}
		\prod_{\alpha\in\cZ(\Le)}\left(1 - \frac{1}{q^{\ell}}\right).
	\end{align*}
	Plugging the above in the decomposition of $\P_{\HYBpd}( \cF )$ and following the
	proof of Theorem~\ref{thm:main2hyb} with $\Lmi,\Li,\Lu$ being replaced by
	$\Lme,\Le,\Lv$ respectively, we conclude the proof of Theorem~\ref{thm:main2poles}.
\end{proof}

	\bibliographystyle{alpha}
	\bibliography{Bib.bib}
	
\end{document}